\documentclass[11pt]{article}
\usepackage[margin = 1in]{geometry}
\usepackage{amsthm,amssymb,amsmath,graphicx,color,mathrsfs,booktabs}
\usepackage[hyphens]{url}
\usepackage{algpseudocode,algorithm,algorithmicx}
\usepackage[colorinlistoftodos,prependcaption,textsize=tiny]{todonotes}
\usepackage[colorlinks = true,
            linkcolor = blue,
            urlcolor  = blue,
            citecolor = blue,
            anchorcolor = blue]{hyperref}
\usepackage{algorithm, algpseudocode, mathbbol} 

\newcommand{\HH}{\mathcal{H}}

\usepackage[numbers,sort]{natbib}

 \bibpunct[, ]{[}{]}{,}{n}{}{,}%

\newcommand{\dmax}{\displaystyle\max}
\newcommand{\dsum}{\displaystyle\sum}

\newcommand{\R}{{\mathbb R}}

\newcommand{\supp}{\operatorname{supp}}
\newcommand{\diag}{\operatorname{diag}}

\renewcommand{\S}{{\mathcal S}}

\newcommand{\X}{{\mathcal X}}
\newcommand{\M}{{\mathcal M}}

\newcommand{\N}{\mathbb{N}}

\newcommand{\vertii}[1]{{\vert\kern-0.25ex\vert\kern-0.25ex\vert #1     \vert\kern-0.25ex\vert\kern-0.25ex\vert}}
\usepackage{enumitem}

\DeclareMathOperator*{\argmin}{arg\,min}
\DeclareMathOperator*{\argmax}{arg\, max}

\DeclareMathOperator{\st}{s.t.}

\newtheorem{theorem}{Theorem}
\newtheorem{proposition}{Proposition}

\newtheorem{corollary}{Corollary}

\newtheorem{remark}{Remark}

\newcommand{\G}{\mathcal G}

\newcommand{\tr}{^\intercal}

\title{Characterization of QUBO reformulations for the maximum $k$-colorable subgraph problem}
\author{
Rodolfo Quintero\thanks{Industrial and Systems Engineering, Lehigh University, USA {\tt roq219@lehigh.edu}} 
\and David Bernal\thanks{Department of Chemical Engineering,
Carnegie Mellon University, USA, {\tt bernalde@cmu.edu}}
\and Tam\'as Terlaky\thanks{Department of Industrial and Systems Engineering, Lehigh University, USA, {\tt
terlaky@lehigh.edu}}
\and Luis F. Zuluaga\thanks{Department of Industrial and Systems Engineering, Lehigh University, USA, {\tt
luis.zuluaga@lehigh.edu}}}

\begin{document}

\maketitle

\abstract{
Quantum devices can be used to solve constrained combinatorial optimization (COPT) problems thanks to the use of penalization methods to embed the COPT problem's constraints in its objective to obtain a quadratic unconstrained binary optimization (QUBO) reformulation of the COPT. However, the particular way in which this penalization is carried out, affects the value of the penalty parameters, as well as the number of additional binary variables that are needed to obtain the desired QUBO reformulation. In turn, these factors substantially affect the ability  of quantum computers to efficiently solve these constrained COPT problems. This efficiency is key towards the goal of using quantum computers to solve constrained COPT problems more efficiently than with classical computers. Along these lines, we consider an important constrained COPT problem; namely, the maximum $k$-colorable subgraph (M$k$CS) problem, in which the aim  is to find an induced $k$-colorable subgraph with maximum cardinality in a given graph. 
This problem arises in channel assignment in spectrum sharing networks, VLSI design, human genetic research, and cybersecurity.
We derive two QUBO reformulations for the M$k$CS problem, and fully characterize the range of the penalty parameters that can be used in the QUBO reformulations. Further, one of the QUBO reformulations of the  M$k$CS problem is obtained without the need to introduce additional binary variables. To illustrate the benefits of obtaining and characterizing these QUBO reformulations, we benchmark different QUBO reformulations of the M$k$CS problem by performing numerical tests on D-Wave's quantum annealing devices. These tests also illustrate the numerical power gained by using the latest D-Wave's quantum annealing device.

\section{Introduction}
\label{sec.intro}

Quantum computing (QC) harnesses the properties physical systems described by quantum mechanics (e.g., subatomic particles) to perform computations in a fundamentally different way than classical computing~\citep{canpress}. It is widely established that QC can, in the future, revolutionize the way we perform and think about computation, and be the backbone of thrilling new technologies and products~\citep{press,canpress,nytimes}. 

In particular, QC has the potential to radically transform our capability to solve  difficult optimization problems for which no traditional numerical or theoretical efficient solution algorithms are known to exist~\citep{montanaro2016quantum}. This is particularly the case for {\em combinatorial optimization} (COPT) problems; that is, optimization problems that are formulated with the use of discrete (e.g., binary) decision variables~\citep{conforti2014integer}. A large number of COPT problems are known to be NP-Hard~\citep[see, e.g.,][]{karp1972reducibility}; that is, there is no known polynomial-time algorithm that can be used to solve them. A very representative problem in this class of COPT NP-Hard problems is the {\em Ising model}~\citep[see, e.g.,][]{brush1967history, singh2020ising,cipra2000ising}. Since its inception, the Ising model has been used to address problems arising in different physical systems (e.g., magnetism, lattice gas, spin glasses), as well as in neuroscience and socio-economics. 

The Ising model belongs to the class of {\em quadratically unconstrained binary optimization} (QUBO) problems~\citep[see,e.g.,][]{pajouh2013characterization}.  Moreover,
both quantum annealing devices~\citep[see, e.g.,][]{lucas2014ising, choi2008minor, johnson2011quantum}, and algorithms (such as the quantum approximate optimization algorithm (QAOA)) for gate-based quantum computers~\citep[see, e.g.,][]{farhi2016quantum,wang2018quantum} are able to address the solution of QUBO problems. This allows the use of quantum technology to solve problems such as the Ising model and the max-cut problem, which has a natural QUBO reformulation~\citep[see, e.g.,][]{farhi2014quantum, king2014algorithm}.
Moreover, quantum technology can be used to solve a broader class of constrained COPT problems that do not have a natural QUBO reformulation. This is due to the fact that penalization methods can be used to embed the COPT problem's constraints in its objective to obtain a QUBO reformulation of the problem.

For some COPT feasibility  problems (i.e., without an objective) that can be formulated using linear equality constraints, the desired QUBO reformulation can be obtained using any positive penalty parameter (to penalize the constraints' violations). For example, consider the QUBO reformulations of the number partitioning problem~\citep{lucas2014ising, nannicini2019performance}, the graph isomorphism problem~\citep{calude2017}, the exact cover problem~\citep{lucas2014ising}, and some planning problems~\citep{rieffel2015case}, to name a few. However, when the COPT problem formulation requires (or uses) nonlinear constraints and/or an objective function, the desired QUBO reformulation is only guaranteed to be obtained for 
values of the penalty parameter(s) that are larger than a known, and potentially large, lower bound. For example, consider the QUBO reformulations for the maximum clique problem~\citep{lucas2014ising}, the traveling salesman problem~\citep{lucas2014ising, nannicini2019performance}, and the minimax matching problem~\citep{lucas2014ising}. Worst, in some cases, the desired QUBO reformulation is only guaranteed to be obtained for an unknown large enough
value of the penalty parameter(s). For example, consider the QUBO reformulations of the job shop scheduling problem~\citep{venturelli2016job}, the de-conflicting optimal trajectories problem~\citep{stollenwerk2019quantum}, the traveling salesman problem with time windows~\citep{papalitsas2019qubo}, and some of the problems discussed in~\citep{glover2019quantum}. Additionally, when the COPT problem formulation requires (or uses) linear inequality constraints, a potentially large number of auxiliary (i.e., slack) binary variables need to be introduced to obtain the desired QUBO reformulation. For example, consider the maximum clique QUBO reformulation provided in~\citep{lucas2014ising}, and the COPT problems considered in~\citep{vyskovcil2019embedding}.

The fact that large (or unknowingly large) penalty parameters, and additional binary variables might be needed to obtain the desired QUBO reformulation can hinder the ability  of quantum computers to  more efficiently solve COPT problems~\citep[see, e.g.,][]{ruan2020quantum, vyskovcil2019embedding,fowler2017improved}. As the results in~\citep{guerreschi2019qaoa} highlight, this efficiency is key towards the goal of using {\em noisy intermediate scale quantum} (NISQ) devices to solve COPT problems more efficiently than with classical computers. Not surprisingly, recent articles look beyond obtaining QUBO reformulations of COPT problems such as the graph isomorphism problem as well as tree and cycle elimination problems, to look for {\em improved}   QUBO reformulations of these problems for NISQ devices~\citep[see, e.g.,][]{calude2017, hua2020improved, fowler2017improved, verma2020optimal, verma2020penalty}. That is, QUBO reformulations that are tailored to be more efficiently used in NISQ devices.

Along these lines,  we consider an important COPT problem; namely, the {\em maximum $k$-colorable subgraph} (M$k$CS) problem~\citep[see, e.g.,][]{kuryatnikova2020maximum}, in which the aim  is to find an induced $k$-colorable subgraph
with maximum cardinality in a given graph. 
This problem arises in channel assignment in spectrum sharing networks (e.g., Wi-Fi or cellular)~\citep{subramanian2007fast,halldorsson2010spectrum}, VLSI design~\citep{fouilhoux2012solving}, human genetic research~\citep{lippert2002algorithmic, fouilhoux2012solving}, telecommunications~\citep{lovasz1979shannon}, and cybersecurity~\citep{berman1990distributed}.

We derive two QUBO reformulations of the M$k$CS problem. The first one is obtained from the standard formulation of the M$k$CS problem in which all the constraints are linear, except for the binary variable constraints. This QUBO reformulation is an improved version of the QUBO reformulation that would be obtained by using the QUBO reformulation approach of~\citet{lasserre2016max}
for this ``linear'' formulation of the M$k$CS. The reason for this is that we characterize the minimum penalization coefficients that can be used to guarantee that the desired QUBO problem, obtained by penalizing the problem's linear constraints violations, is indeed equivalent to the original problem. 
Furthermore, we characterize the equivalence of the QUBO reformulation not only in terms of the objective value, but also in terms of the optimal solution obtained from this QUBO reformulation. In particular, we find that when the minimal values of the penalization coefficients are used, the QUBO reformulation is equivalent to the M$k$CS in terms of the problems' objectives, but not in terms of the problems' optimal solutions. However, we show that in this case, the QUBO reformulation's optimal solution can be used, in a simple way, to obtain the M$k$CS problem's optimal solution. In what follows, we will refer to this QUBO reformulation of the M$k$CS problem as the {\em linear-based} QUBO reformulation.

The second QUBO reformulation of the M$k$CS problem is obtained from a formulation of the M$k$CS problem in which all the linear constraints are first formulated as nonlinear equality constraints. Analogous to the results obtained for the linear-based QUBO reformulation of the M$k$CS problem, we derive a {\em nonlinear-based} QUBO reformulation of the M$k$CS problem. Then, we characterize the minimum penalizations coefficients that can be used to guarantee that the desired nonlinear-based QUBO problem, 
obtained by penalizing the problem's linear constraints violations, is indeed equivalent to the original problem. Furthermore, we characterize the equivalence of the nonlinear-based QUBO reformulation not only in terms of the objective value, but also in terms of the optimal solution obtained from this nonlinear-based QUBO reformulation. In particular, we find that when the minimal values of the penalization coefficients are used, the nonlinear-based QUBO reformulation is equivalent to the M$k$CS in terms of the problems' objectives, but not in terms of the problems' optimal solutions. However, we show that in this case, the nonlinear-based QUBO reformulation's optimal solution can be used, in a simple way, to obtain the M$k$CS problem's optimal solution. This latter result extends the work done in characterizations of QUBO reformulations of the {\em stable set problem}~\citep{harant2000some, abello2001finding, boros2007local}, which is equivalent to the M$k$CS problem when $k=1$.
The nonlinear-based QUBO reformulation of the M$k$CS problem is a substantial improvement over the linear-based QUBO reformulation of the M$k$CS problem, in significant part, because the former QUBO does not need the addition of any auxiliary (i.e., slack) binary variables beyond the ones that define the original problem's formulation.

To illustrate the benefits of obtaining and characterizing these QUBO reformulations, we benchmark different QUBO reformulations of the M$k$CS problem using a quantum annealing device, and in particular, we look at how embedding requirements and theoretical and numerical convergence rates change depending on the QUBO reformulation being used, as well as the parameters with which is used.

The rest of the article is organized as follows. In Section~\ref{sec:prelim}, we present some relevant discussion to motivate our work, as well as results about QUBO reformulations for COPT problems. In Section~\ref{sec:MkCS}, we formally present the M$k$CS problem and two associated QUBO reformulations. The first one, in Section~\ref{sec:linear}, is based on a ``linear'' (modulo the binary variable constraints) formulation of the M$k$CS problem. The second one, in Section~\ref{sec:nonlinear}, is based on a ``nonlinear'' (beyond the binary variable constraints) formulation of the M$k$CS problem. In Section~\ref{sec:benchmark}, we benchmark these two QUBO reformulation by performing numerical tests on D-Wave's quantum annealing devices. We also illustrate the numerical power gained by using the latest D-Wave's quantum annealing devices.  In Section~\ref{sec:end}, we finish with some concluding remarks.

\section{Preliminaries}
\label{sec:prelim}

Formally, given a set of $n$ binary decision variables $x \in \{0,1\}^n$ (or $x \in \{-1,1\}^n$) when appropriate), a vector $f \in \R^n$, and a matrix $Q \in \S^n$, where~$\S^n$ is the set of symmetric matrices in $\R^{n\times n}$, a {\em quadratically unconstrained binary optimization} (QUBO) problem is the problem of finding~\citep[see, e.g.,][]{pajouh2013characterization, boros2007local}:
\begin{equation}
\label{eq:qubo}
\tag{QUBO}
\begin{array}{llllllll}
z^* = & \min & x\tr Q x + f\tr x \\
         & \st & x \in \{0,1\}^n.
\end{array}
\end{equation}
It is well-known that the Ising model belongs to the class of QUBO problems (using $\{-1,1\}$ binary variables)~\citep[see, e.g.,][]{lucas2014ising,fowler2017improved}. 
Moreover, other distinguished NP-Hard COPT problems can be naturally formulated, or easily reformulated as a QUBO problem. Foremost among this type of problems is the max-cut problem~\citep[see, e.g.,][]{goemans1995improved}, which arises in multiple  important applications in science and engineering~\citep[see, e.g.,][Sec. 6]{poljak1995maximum}.
Given an undirected graph $G(V,E)$, the aim in the max-cut problem is to find a subset of nodes (or cut) $S \subseteq V$, such that the cardinality of the set of edges in $E$ between the nodes in $S$ and $S^c := V \setminus S$ is maximized. The max-cut problem can be naturally formulated (disregarding objective constants) as a \ref{eq:qubo}  problem (using $\{-1,1\}$ binary variables) by letting $Q = A$, $f=0$, where $A \in \R^{V\times V}$ is the node-to-node adjacency matrix of $G(V,E)$, or by setting
$Q=-\diag(Ae)+2A$ and $f=0$ (using  $\{0,1\}$ binary variables).

Thanks to the \ref{eq:qubo}  reformulation of the max-cut problem, the ability of quantum computers to solve the max-cut problem has been widely studied in the literature. For example, consider the use 
QAOA algorithms in~\citep{farhi2016quantum, wang2018quantum, crooks2018performance}, and of quantum annealing devices in~\citep{king2015benchmarking, king2014algorithm} to solve instances of the max-cut problem. Furthermore, QUBO reformulations can be obtained for a broader class of COPT problems that do not have a natural QUBO reformulation. This is done by using penalization methods to embed the COPT problem's constraints in its objective~\citep[see, e.g.,][to name just a few]{lucas2014ising, nannicini2019performance, calude2017, fowler2017improved, rieffel2015case, venturelli2016job, stollenwerk2019quantum, glover2019quantum, vyskovcil2019embedding}. This approach clearly broadens the class of COPT problems that can be addressed with NISQ devices. However, the efficacy of NISQ devices to solve this broader class of COPT problems can be highly affected by the way in which the corresponding QUBO reformulation is obtained. This is because the performance of NISQ devices is highly affected by the number of qubits and the coefficients that are required to encode a QUBO~\citep[see, e.g.,][]{calude2017, hua2020improved, fowler2017improved}.

To illustrate this fact, consider the problem of obtaining a QUBO reformulation for the {\em maximum clique} problem. Given an undirected graph $G(V, E)$, the aim in the maximum clique problem is to find the set of nodes $S \subseteq V$ with the highest cardinality such that the graph induced by $S$ is a clique; that is, a complete subgraph~\citep[see, e.g.,][]{pardalos1994maximum}. The cardinality of the largest induced clique of $G$ is referred to as the clique number $\chi(G)$. \citet[][Sec. 2.3]{lucas2014ising} obtains a QUBO reformulation for the maximum clique problem by first noticing that $G(V, E)$ contains a clique of size $K \in \{2,\dots,|V|\}$ (i.e., w.l.o.g. assume $|E| \ge 1$) if and only if there is $x \in \{0,1\}^{|V|}$ such that $\sum_{i=1}^{|V|}x_i = K$, and $\sum_{(i,j) \in E} x_ix_j = \frac{1}{2}K(K-1)$. Thus, the maximum clique problem can be formulated as $\chi(G) = \max\{K \in \{2,\dots, |V|\}: \sum_{i=1}^{|V|}x_i = K, \sum_{(i,j) \in E} x_ix_j = \frac{1}{2}K(K-1),x \in \{0,1\}^{|V|}\}$. Furthermore, \citet[][Sec. 2.3]{lucas2014ising} shows that this latter problem can be reformulated as the following QUBO.
\begin{equation}
\label{eq:lucasclique}
\begin{array}{lllllll}
\chi(G) =  & \min & - \dsum_{i=1}^{|V|}x_i + (\Delta + 2)\left(1-\dsum_{k=2}^\Delta y_k \right)^2 + 
(\Delta + 2)\left(\dsum_{k=2}^\Delta ky_k  - \dsum_{i=1}^{|V|}x_i \right)^2 + \\
 & & \frac{1}{2}\left ( \dsum_{k=2}^\Delta ky_k \right ) \left ( -1  + \dsum_{k=2}^\Delta ky_k \right ) - \dsum_{(i,j) \in E} x_ix_j \\
 & \st & x \in \{0,1\}^{|V|}, y_k \in \{0,1\}, k=2,\dots,\Delta,
\end{array}
\end{equation}
where $\Delta$ is the degree of $G(V,E)$, and the auxiliary variable $y_k =1$ if $\chi(G) = k$ and $y_k = 0$ otherwise for $k=2,\dots,\Delta$. Note that the QUBO problem~\eqref{eq:lucasclique} uses $|V| + \Delta$ logical qubits and {\em coefficients}~\citep[][Section 1.2]{vyskovcil2019embedding} that belong to the range $[-2\Delta(\Delta+2), 2\Delta^3+ 3\Delta(\Delta-1)+4]$ (after disregarding constant terms and appropriately replacing $x_i \to x_i^2$, $i=1,\dots,|V|$, $y_k \to y_k^2$, $k=2,\dots,\Delta$ in the objective of~\eqref{eq:lucasclique} to make it a homogenous quadratic). The performance of NISQ devices on solving QUBO problems is negatively affected by the use of a larger number of logical qubits and larger coefficients~\citep[see, e.g.,][]{vyskovcil2019embedding, glover2019quantum, calude2017, hua2020improved, fowler2017improved}. In this context, it is natural to ask if there are {\em improved}~\citep[see, e.g.,][]{calude2017, hua2020improved, fowler2017improved, verma2020optimal, verma2020penalty} QUBO reformulation for the maximum clique problem. For example, notice that by slightly changing the definition and number of the auxiliary variables in~\eqref{eq:lucasclique}, the range of the coefficients used in~\eqref{eq:lucasclique} can be substantially reduced. Namely, let $y \in \{0,1\}^{\Delta}$ be defined by $\sum_{k=1}^{\Delta} y_k = K$ if $\chi(G) = K$ for $K \in \{1,\dots,\Delta\}$. Then the maximum clique problem is equivalent to:
\begin{equation}
\label{eq:lucascliquemod}
\begin{array}{lllllll}
\chi(G) = & \min & - \dsum_{i=1}^{|V|}x_i + 
(\Delta + 2)\left(\dsum_{k=1}^\Delta y_k  - \dsum_{i=1}^{|V|}x_i \right)^2 + \\
&  & \frac{1}{2}\left ( \dsum_{k=1}^\Delta y_k \right ) \left ( -1  + \dsum_{k=1}^\Delta y_k \right ) - \dsum_{(i,j) \in E} x_ix_j \\
 & \st & x \in \{0,1\}^{|V|}, y \in \{0,1\}^{\Delta}.
\end{array}
\end{equation}
Note that the QUBO problem~\eqref{eq:lucascliquemod} uses coefficients that belong to a much smaller range $[-2(\Delta+2), 4(\Delta+2)+1]$ than the range of coefficients used in the QUBO problem~\eqref{eq:lucasclique} (after disregarding constant terms and appropriately replacing $x_i \to x_i^2$, $i=1,\dots,|V|$, $y_k \to y_k^2$, $k=1,\dots,\Delta$ in the objective of~\eqref{eq:lucasclique} to make it an homogenous quadratic). However, a much better QUBO formulation for the maximum clique problem can be obtained by using the fact that $\chi(G) = \alpha(G^c)$~\citep[see, e.g.,][]{pardalos1994maximum}, where for a graph~$G(V, E)$, $G^c = G(V, E^c)$ is the complement of $G$, and $\alpha(G)$ stands for the {\em stable set number} of the graph~$G$~\citep[see, e.g.,][]{harant2000some}; that is, the size of the largest cardinality set $S \subseteq V$, such that there are no edges between the nodes in $S$. This fact can be used to show that (see, e.g., \citep[][Thm. 6]{calude2017} or \citep[][Thm. 2.3]{pardalos1994maximum}, among others)
\begin{equation}
\label{eq:easyclique}
\chi(G) = \alpha(G^c)= \min \left \{- \dsum_{i=1}^{|V|}x_i + 2\dsum_{(i,j) \not \in E} x_ix_j: x \in \{0,1\}^{|V|} \right \}
\end{equation}
Note that the QUBO problem~\eqref{eq:easyclique} uses $|V|$ logical qubits and coefficients that belong to the range $\{-1, 2\}$. Thus, in terms of number of logical qubits and range of the coefficients used in the QUBO reformulation,~\eqref{eq:easyclique} improves both~\eqref{eq:lucascliquemod} and~\eqref{eq:lucasclique}. It is worth pointing out that the QUBO reformulation~\eqref{eq:easyclique} has been stated in numerous articles~\citep[see, e.g.,][to name a few]{chapuis2017finding, abello2001finding, yarkoni2018first, pardalos1994maximum}. Moreover, it is well known that the range of the coefficients in~\eqref{eq:easyclique} can be further reduced to $\{-1,1\}$. Namely, it has been proved (or stated) in numerous articles~\citep[see, e.g.,][]{nannicini2019performance, wocjan20032, boros2007local, harant2000some, abello2001finding, pajouh2013characterization} that

\begin{equation}
\label{eq:bestclique}
\chi(G) = \alpha(G^c)= \min \left \{- \dsum_{i=1}^{|V|}x_i + \dsum_{(i,j) \not \in E} x_ix_j: x \in \{0,1\}^{|V|} \right \}
\end{equation}

There is, however, a caveat in the QUBO reformulation~\eqref{eq:bestclique}. For any $x \in \R^n$, let $\supp(x) = \{i \in \{1, \dots,n\}: x_i \neq 0\}$. Unlike for~\eqref{eq:lucasclique}--\eqref{eq:easyclique}, given $
x^* \in \argmin\{\eqref{eq:bestclique}\}$, $\supp(x^*)$ might not be a clique on $G$ (nor an independent set in $G^c$). That is, while the QUBO problems~\eqref{eq:lucasclique}--\eqref{eq:easyclique} are equivalent to the maximum clique problem in terms of both objective value and (loosely speaking) optimal solution, in general, the QUBO problem~\eqref{eq:bestclique} is equivalent to the maximum clique problem {\em only} in terms of objective value. This important topic will be revisited and discussed in detail in Section~\ref{sec:nonlinear}.

Along these lines, in what follows, we consider the problem of obtaining not only a QUBO reformulation, but improved QUBO reformulation of a keystone COPT problem; 
namely, the {\em maximum $k$-colorable subgraph} (M$k$CS) problem~\citep[see, e.g.,][]{kuryatnikova2020maximum}.

\section{The $k$-subgraph coloring problem}
\label{sec:MkCS}

Let $k\geq 1$ colors and a graph $G=(V,E)$ on $n$ vertices be given. A subgraph $H$ of $G$ is $k$-colorable if we can assign to each vertex of $H$ a color such that no two adjacent vertices in $H$ have the same color. The maximum  $k$-colorable subgraph problem (M$k$CS) aims at 
finding a $k$-colorable subgraph $H$ of $G$ with maximum cardinality. To model this problem, notice that any $k$-coloring of a subgraph of $G$ can be encoded in the following way. For any $i\in [n]$ (where for any $t \in \N$, $[t]:=\{1,\dots,t\}$) and $r\in [k]$, let
    \begin{equation} \label{Matrix k coloring}   
 x_{ir} = \begin{cases}
1, & \text{ if vertex } i \in [n] \text{ is colored with color } r \in [k],\\
0, & \text{otherwise}. 
\end{cases}
\end{equation}
Then,  $x\in\{0,1\}^{n\times k}$ defines a $k$-coloring of a subgraph of $G$ if and only
\begin{align}
\label{eq:linearcons}
\begin{split}
 x_{ir} + x_{jr} \leq 1, & \text{ for all } (i,j)\in E, r\in [k],\\
    \sum_{r\in [k]} x_{ir} \leq 1, & \text{ for all } i \in [n].   
\end{split}
\end{align}
Then, the M$k$CS can be formulated as~\cite[see, e.g.,][]{kuryatnikova2020maximum}: 
\begin{equation}
\label{eq:MkCS}
\begin{array}{lllllll}
    \alpha_k(G):= & \dmax_{x\in \{0,1\}^{n\times k}} & \dsum_{i\in [n], r\in [k]} x_{ir} \\[2ex]
    & \st & x_{ir} + x_{jr} \leq 1, & \text{for all } (i,j) \in E, r\in [k],\\
     & & \dsum_{r\in [k]} x_{ir} \leq 1, & \text{for all } i \in [n]. 
\end{array}
\end{equation}
The M$k$CS problem falls into the class of NP-complete problems~\citep{yannakakis1987maximum}. Moreover, even approximating this problem is known to be NP-hard~\citep{lund1993approximation}. For $k=1$, the  M$k$CS is equivalent to the maximum stable set problem (i.e., $\alpha_1(G) = \alpha(G)$) that has been widely and thoroughly studied in the literature; and in particular, in the quantum computing literature~\citep[see, e.g.,][]{nannicini2019performance, wocjan20032, chapuis2017finding}. The cases $k=2$, which is also referred to as the maximum bipartite subgraph problem,
and $k>2$ are considered significantly less in the literature~\citep[see][for details]{kuryatnikova2020maximum}. However, as mentioned earlier, the M$k$CS problem arises in channel assignment in spectrum sharing networks (e.g., Wi-Fi or cellular)~\citep{subramanian2007fast,halldorsson2010spectrum}, VLSI design~\citep{fouilhoux2012solving}, human genetic research~\citep{lippert2002algorithmic, fouilhoux2012solving}, telecommunications~\citep{lovasz1979shannon}, and cybersecurity~\citep{berman1990distributed}. Thus, a range
of approaches have been studied in the literature to address the solution of the M$k$CS problem, for example, using semidefintie optimization techniques~\citep[see, e.g.,][]{van2016new, kuryatnikova2020maximum} or integer programming techniques~\citep[see, e.g.,][]{campelo2010combined, januschowski2011maximum, calude2017}. 

Next, we obtain and characterize QUBO reformulations for the M$k$CS problem that allow to address its solution using quantum technology. Before presenting these results, let us mention some additional facts about the M$k$CS problem that will be relevant to the discussion in what follows. 

Notice that 
a M$k$CS $H$ of $G(V,E)$ can be recovered from any $x^* \in \argmax\{ \alpha_k(G)\}$; that is, $H:=G(V_H, E_H)$, where $V_H = \{i \in [n]: x^*_{ir} > 0\text{ for some }r \in [k]\}$, $E_H:= \{(i,j) \in E: i, j \in V_H\}$, and the coloring of the vertices is obtained by coloring vertex~$i \in V_H$ with color $r \in [k]$ if and only if $x^*_{ir} = 1$. Furthermore, given $\tilde{x} \in \{0,1\}^{n \times k}$, it is very simple to obtain a feasible solution $x' \in \{0,1\}^{n \times k}$ for the M$k$CS problem by sequentially dropping color $r' \in [k]$ from vertex $i' \in [n]$; that is, setting $\tilde{x}_{i'r'} = 0$, if $\tilde{x}_{i'r'} = 1$ and there exists $(i',j) \in E$ such that $\tilde{x}_{i'r'}+\tilde{x}_{jr'} >1$ or 
$\sum_{r\neq r'} \tilde{x}_{i'r} \ge 1$. This simple fact is formally stated in Algorithm~\ref{alg:AlgorithmA}, in a particular form that will be helpful in stating some of the QUBO characterization results that follow.

\begin{figure}[htb]
\centering
\begin{minipage}{0.7\linewidth}
\begin{algorithm}[H]
\small
\caption{M$k$CS feasibility}
\begin{algorithmic}[1]
\State  \textbf{Input}{ $k \ge 1$, $G(V,E)$, $|V| = n$, $x \in \{0,1\}^{n \times k}$}
\For{$i\in [n]$, $(i,j) \in E$, $r \in [k]$}
	\If{$x_{ir} + x_{jr} > 1$}
	\State $x_{ir} \to 0$ \label{step:4}
	\EndIf
\EndFor
\For{$i\in [n]$, $r \in [k]$}
\If{$x_{ir} = 1$ and $\sum_{p  \neq r \in [k]}x_{ip} \ge 1$}
\State  $x_{ir} \to 0$ \label{step:9}
\EndIf
\EndFor
\State  \textbf{Output}{ $x' :=x$ a feasible solution for the M$k$SC problem}
\end{algorithmic}
\label{alg:AlgorithmA}
\end{algorithm}
\end{minipage}
\end{figure}

\subsection{Linear-based QUBO reformulation} 
\label{sec:linear}

Based on the formulation~\eqref{eq:MkCS} of the M$k$CS problem in which all the constraints, except for the binary variable constraints are {\em linear}, we can derive and characterize a {\em linear-based} QUBO reformulation for the M$k$CS problem. For that purpose, let us first introduce some notation. 

Given $k\geq 1$, a graph $G=(V,E)$ on $n$ vertices, and $x \in \{0,1\}^{n\times k}$, $s \in \{0,1\}^{|E| \times k}$, $t \in \{0,1\}^n$, let
\begin{equation}
\label{eq:h0}
H_0(x) =  \dsum_{i \in [n], r \in [k]} x_{ir}^2,
\end{equation}
and
\begin{subequations}
\label{eq:h1h2l}
\begin{align}
H_1^l(x, s) &= \dsum_{(i,j) \in E, r \in [k]} \left(x_{ir}+x_{jr}+s_{ijr}-1\right)^2, \label{eq:h1l}\\
H_2^l(x, t) & = \dsum_{i\in[n]} \left( \sum_{r\in [k]}x_{ir}+t_i-1\right)^2  \label{eq:h2l}.
\end{align}
\end{subequations}
Furthermore, we define the following simple mappings. Given $x \in \{0,1\}^{n \times k}$  and $i' \in [n], r'\in [k]$, let the mapping $\X_{i'r'}(x):  \{0,1\}^{n \times k} \to  \{0,1\}^{n \times k}$ 
be defined by
\begin{equation}
\label{eq:mapx}
x_{ir} \to \left \{ \begin{array}{ll} 
                                    0 & \text{if } i=i', r=r'\\
                                    x_{ir} & \text{otherwise}\\
                                    \end{array} \right . , i \in [n], r \in [k].
                                    \end{equation}
Note that $\X_{i'r'}(x)$ is a generalization of the mapping used on proofs regarding QUBO reformulations of the stable set number problem (i.e., M$1$CS)~\citep[see, e.g.,][]{boros2007local, harant2000some, abello2001finding, pajouh2013characterization,wocjan20032}. Here, however, to deal with the general case $k>1$, we need an additional mapping.

Given $p = \{0,1\}$, $s \in \{0,1\}^{|E| \times k}$, $t \in \{0,1\}^{n}$, and $(i',j') \in E$, $r' \in [k]$, let the mapping $\M^p_{i'j',r'}(s,t):  \{0,1\}^{|E| \times k + n} \to   \{0,1\}^{|E| \times k + n} $ be defined by
\begin{subequations}
\label{eq:mapst}
\begin{align}
s_{ijr} &\to \left \{ \begin{array}{ll} 
                                    1-s_{i'jr'} & \text{if } i=i', j \neq j', r=r'\\
                                    (1-s_{i'jr'})p             & \text{if } i=i', j = j', r=r'\\
                                    s_{ijr} & \text{otherwise}\\
                                    \end{array} \right . , (i,j) \in E, r \in [k],\label{eq:mapst_s}\\[2ex]
t_{i} &\to \left \{ \begin{array}{ll} 
                                    1-p & \text{if } i=i',\\
                                    t_{i} & \text{otherwise}\\
                                    \end{array} \right . , i \in [n]. \label{eq:mapst_t}
\end{align}
\end{subequations}               

With these definitions in hand, we can now obtain the desired linear-based QUBO reformulation of the M$k$CS problem. For any $c_1, c_2 > 0$ define the QUBO problem:
\begin{equation}
\label{eq:qubolinear}
\begin{array}{llll}
Q^l_{c_1,c_2}(k,G) := & \max & H^l_{c_1,c_2}(x,s,t):= H_0(x) - c_1H_1^l(x,s) - c_2H_2^l(x,t)\\[2ex]
   & \st & x \in \{0,1\}^{n\times k}, s\in\{0,1\}^{|E|\times k}, t\in \{0,1\}^n.
   \end{array}
\end{equation}
 
\begin{theorem}[linear-based QUBO reformulation of M$k$CS problem]
\label{thm:linearQUBO}
Let $k\geq 1$ and a graph $G=(V,E)$ on $n$ vertices be given.  Then, for any $c_1 >1, c_2 > 1$, $Q^l_{c_1,c_2}(k,G) = \alpha_k(G)$, and if $\tilde{x} \in \argmax_x\{Q^l_{c_1,c_2}(k,G)\}$ then $\tilde{x} \in \argmax\{\alpha_k(G)\}$.
\end{theorem}

\begin{proof}
First, notice that $\tilde{x}$ is well defined and $Q^l_{c_1,c_2}(k,G)$ is attained as~\eqref{eq:qubolinear} is defined over a compact feasible set. Also, notice that for any $c_1, c_2 > 0$ and any feasible solution $x' \in \{0,1\}^{n\times k}$ for the M$k$CS problem~\eqref{eq:MkCS} with objective value $z(x'):= \sum_{i\in [n], r\in [k]} x_{ir}$, one can construct a feasible solution for~\eqref{eq:qubolinear}; that is,  $x = x'$, $s_{ijr} = 1-x'_{ir}-x'_{jr}$, for all $(i,j) \in E, r\in [k]$, and $t_i = 1 - \sum_{r \in [k]} x'_{ir}$, with objective value $H^l_{c_1,c_2}(x,s,t) = z(x')$. Thus, if $c_1, c_2 > 0$, the QUBO problem~\eqref{eq:qubolinear} is a  relaxation of~\eqref{eq:MkCS}, and consequently  $Q^l_{c_1,c_2}(k,G) \ge \alpha_k(G)$. Thus, to prove the result, it is enough to show that when $c_1, c_2 > 1$, one has that~$\tilde{x}$ is a feasible solution for~\eqref{eq:MkCS}. By contradiction, assume this is not the case and let $c_1, c_2 > 1$, $(\tilde{s}, \tilde{t}) :=\argmax_{(s,t)}\{Q^l_{c_1,c_2}(k,G)\}$. Then either: (1) there is at least an $(i',j') \in E$ and $r'\in [k]$ such that $\tilde{x}_{i'r'} + \tilde{x}_{j'r'} >1$; or (2) there is at least an $i' \in [n]$ and $r' \in [k]$ such that $\tilde{x}_{i'r'} = 1$ and $\sum_{r\neq r' \in [k]} \tilde{x}_{i'r} \ge1$. 

For case~(1), consider the feasible solution $(x, s^0, t^0) \in \{0,1\}^{n\times k + |E|\times k + n}$ for~\eqref{eq:qubolinear} obtained from $(\tilde{x}, \tilde{s}, \tilde{t})$ by letting $(x,s,t) = (\X_{i'r'}(\tilde{x}), \M^0_{i'j',r'}(\tilde{s}, \tilde{t}))$ (cf.,~\eqref{eq:mapx},~\eqref{eq:mapst}). It then follows 
from~\eqref{eq:h0},~\eqref{eq:mapx}, and the fact that $\tilde{x}_{i'r'}=1$ that 
\begin{equation}
\label{eq:h0diffeq}
H_0(x) =  H_0(\tilde{x}) - 1.
\end{equation}
Also, from~\eqref{eq:h1l},~\eqref{eq:mapx},~\eqref{eq:mapst_s}, and the fact that $\tilde{x}_{i'r'}=\tilde{x}_{j'r'}=1$, it follows that $-H^l_1(x,s^0) = - H^l_1(\tilde{x},\tilde{s}) + \sum_{(i',j\neq j') \in E} 4\tilde{x}_{jr'}\tilde{s}_{i'jr'} + (1+\tilde{s}_{i'j'r'})^2$. Thus,
\begin{equation}
\label{eq:h1ldiff1}
-H^l_1(x,s^0) \ge - H^l_1(\tilde{x},\tilde{s}) + 1.
\end{equation}
Further, from~\eqref{eq:h2l},~\eqref{eq:mapx},~\eqref{eq:mapst_t}, and the fact that $\tilde{x}_{i'r'}=1$, it follows that $-H^l_2(x,t^0) = - H^l_2(\tilde{x},\tilde{t}) + 2\tilde{t}_{i'}\sum_{r \neq r' \in [k]} \tilde{x}_{i'r} + \tilde{t}^2_{i'}$. Thus,
\begin{equation}
\label{eq:h2ldiff1}
-H^l_2(x,t^0) \ge - H^l_2(\tilde{x},\tilde{t}).
\end{equation}
Using~\eqref{eq:h0diffeq},~\eqref{eq:h1ldiff1},~\eqref{eq:h2ldiff1}, it follows that $H^l_{c_1,c_2}(x,s^0,t^0) \ge H^l_{c_1,c_2}(\tilde{x},\tilde{s},\tilde{t}) -1+ c_1 > H^l_{c_1,c_2}(\tilde{x},\tilde{s},\tilde{t}) = Q^l_{c_1,c_2}(k,G)$, which contradicts the optimality of $(\tilde{x},\tilde{s},\tilde{t})$ for~\eqref{eq:qubolinear}. 

We proceed analogously for case (2). Consider the feasible solution $(x, s, t) \in \{0,1\}^{n\times k + |E|\times k + n}$ for~\eqref{eq:qubolinear} obtained from $(\tilde{x}, \tilde{s}, \tilde{t})$ by letting $(x,s^1,t^1) =$ $(\X_{i'r'}(\tilde{x}), \allowbreak \M^1_{i',r'}(\tilde{s}, \tilde{t}))$ (cf.,~\eqref{eq:mapx},~\eqref{eq:mapst}). It then follows 
from~\eqref{eq:h0},~\eqref{eq:mapx}, and the fact that $\tilde{x}_{i'r'}=1$ that~\eqref{eq:h0diffeq} holds. 
Also, from~\eqref{eq:h1l},~\eqref{eq:mapx},~\eqref{eq:mapst_s}, and the fact that $\tilde{x}_{i'r'}=1$, it follows that $-H^l_1(x,s^1) = - H^l_1(\tilde{x},\tilde{s}) + \sum_{(i',j) \in E} 4\tilde{x}_{jr'}\tilde{s}_{i'jr'} $. Thus,
\begin{equation}
\label{eq:h1ldiff2}
-H^l_1(x,s^1) \ge - H^l_1(\tilde{x},\tilde{s}).
\end{equation}
Further, from~\eqref{eq:h2l},~\eqref{eq:mapx},~\eqref{eq:mapst_t}, and the fact that $\tilde{x}_{i'r'}=1$, $\sum_{r \neq r' \in [k]} \tilde{x}_{j'r} \ge 1$, it follows that $-H^l_2(x,t^1) = - H^l_2(\tilde{x},\tilde{t}) + 2(\sum_{r \neq r' \in [k]} \tilde{x}_{i'r})(1+\tilde{t}_{i'}) + \tilde{t}^2_{i'} - 1$. Thus,
\begin{equation}
\label{eq:h2ldiff2}
-H^l_2(x,t^1) \ge - H^l_2(\tilde{x},\tilde{s}) +1.
\end{equation}
Using~\eqref{eq:h0diffeq},~\eqref{eq:h1ldiff2},~\eqref{eq:h2ldiff2}, it follows that $H^l_{c_1,c_2}(x,s^1,t^1) \ge H^l_{c_1,c_2}(\tilde{x},\tilde{s},\tilde{t}) -1+ c_2 > H^l_{c_1,c_2}(\tilde{x},\tilde{s},\tilde{t}) = Q^l_{c_1,c_2}(k,G)$, which contradicts the optimality of $(\tilde{x},\tilde{s},\tilde{t})$ for~\eqref{eq:qubolinear}. 

Therefore $\tilde{x}$ satisfies that there is no $(i',j') \in E$ and $r'\in [k]$ such that $\tilde{x}_{i'r'} + \tilde{x}_{j'r'} >1$, or $i' \in [n]$ and $r' \in [k]$ such that $\sum_{r \in [k]} \tilde{x}_{i'r} >1$. Therefore $\tilde{x}$ is a feasible solution of~\eqref{eq:MkCS}, which finishes the proof.
\end{proof}

It is worth to mention that, loosely speaking, the general form of the QUBO reformulation~\eqref{eq:qubolinear} for the M$k$CS problem can be obtained by using the recent results of~\citet[][Thm. 2.2]{lasserre2016max}. Namely, one can use this result after reformulating the M$k$CS problem constraints as equality constraints using the approach described in~\citep[][Sec. 2.3]{lasserre2016max}. Then, after reformulating the problem using $\{1,-1\}$ binary variables (instead of $\{0,1\}$ binary variables),~\citep[][Thm. 2.2]{lasserre2016max} can be used to obtain a QUBO reformulation of the M$k$CS problem. However, this reformulation would require the use a penalty parameter with a value larger than $nk$ (cf., with the values of $c_1, c_2$ in Theorem~\ref{thm:linearQUBO}), and require the use of more auxiliary (i.e., slack) binary variables than the ones used in Theorem~\ref{thm:linearQUBO}. Thus, Theorem~\ref{thm:linearQUBO} provides an improved QUBO reformulation of the M$k$CS problem than the one that would be obtained using~\citep[][Thm. 2.2]{lasserre2016max}. 

Later, in Section~\ref{sec:linearagain}, we will further characterize the QUBO reformulation~\eqref{eq:qubolinear}
for the M$k$CS. Next, however, we derive and characterize a QUBO reformulation for the M$k$CS in which no auxiliary (i.e., slack) binary variables are needed.

\subsection{Nonlinear QUBO reformulation}
\label{sec:nonlinear}

Next, we obtain an improved QUBO reformulation for the M$k$CS problem in terms of the number of binary decision variables required in the QUBO reformulation, when compared with the one provided  and characterized in Section~\ref{sec:linear}. For this purpose, first notice that for any $x\in\{0,1\}^{n\times k}$, the linear constraints in~\eqref{eq:linearcons} are equivalent to the nonlinear constraints
\begin{align}
\label{eq:nonlinearcons}
\begin{split}
 x_{ir}x_{jr} = 0, & \text{ for all } (i,j) \in E, r\in [k],\\
 x_{ir}x_{ip} = 0, & \text{ for all } i \in [n], (r,p \neq r)\in [k]\times [k].   
\end{split}    
\end{align}
Then, consistent with~\eqref{eq:nonlinearcons}, given $k\geq 1$, a graph $G=(V,E)$ on $n$ vertices, and $x \in \{0,1\}^{n\times k}$, let
\begin{subequations}
\label{eq:h1h2n}
\begin{align}
H_1^n(x) &= \dsum_{(i,j) \in E, r \in [k]} x_{ir}x_{jr}, \label{eq:h1n}\\
H_2^n(x) & = \dsum_{i\in[n]} \left( \sum_{r\in [k],p \neq r\in[k]}x_{ir}x_{ip} \right ) \label{eq:h2n}.
\end{align}
\end{subequations}
With these definitions in hand we can now obtain the desired {\em nonlinear-based} QUBO reformulation of the M$k$CS problem. For any $c_1, c_2 > 0$ define the QUBO problem:
\begin{equation}
\label{eq:qubononlinear}
\begin{array}{llll}
Q^n_{c_1,c_2}(k,G) := & \max & H^n_{c_1,c_2}(x):= H_0(x) - c_1H_1^n(x) - c_2H_2^n(x)\\[2ex]
   & \st & x \in \{0,1\}^{n\times k}.
   \end{array}
\end{equation}

\begin{theorem}[nonlinear-based QUBO reformulation of M$k$CS problem]
\label{thm:nonlinearQUBO}
Let $k\geq 1$ and a graph $G=(V,E)$ on $n$ vertices be given.  Then, 
for any $c_1 >1, c_2 > 1$, $Q^n_{c_1,c_2}(k,G) = \alpha_k(G)$, and 
if $\tilde{x} \in \argmax\{Q^n_{c_1,c_2}(k,G)\}$ then $\tilde{x} \in \argmax\{\alpha_k(G)\}$.
\end{theorem}

\begin{proof}
The proof is mostly analogous to the proof of Theorem~\ref{thm:linearQUBO}.
First, notice that~$\tilde{x}$ is well defined and $Q^n_{c_1,c_2}(k,G)$ is attained as~\eqref{eq:qubononlinear} is defined over a compact feasible set. Also, notice that for any $c_1, c_2 > 0$ and any feasible solution $x' \in \{0,1\}^{n\times k}$ for the M$k$CS problem~\eqref{eq:MkCS} with objective value $z(x'):= \sum_{i\in [n], r\in [k]} x_{ir}$, one hast that $x'$ is a feasible solution of~\eqref{eq:qubononlinear}   with objective value $H^n_{c_1,c_2}(x) = z(x')$ (i.e., $x'$ satisfies~\eqref{eq:nonlinearcons}). Thus, if $c_1, c_2 > 0$, the QUBO problem~\eqref{eq:qubononlinear} is a  relaxation of~\eqref{eq:MkCS}, and consequently  $Q^n_{c_1,c_2}(k,G) \ge \alpha_k(G)$. Thus, 
to prove the result,
it is enough to show that when $c_1, c_2 > 1$, one has that~$\tilde{x}$ is a feasible solution for~\eqref{eq:MkCS}. By contradiction, assume this is not the case and let $c_1, c_2 > 1$. Then either: (1) there is at least an $(i',j') \in E$ and $r'\in [k]$ such that $\tilde{x}_{i'r'} + \tilde{x}_{j'r'} >1$; or (2) there is at least an $i' \in [n]$ and $r' \in [k]$ such that $\tilde{x}_{i'r'} = 1$ and $\sum_{r\neq r' \in [k]} \tilde{x}_{i'r} \ge1$. Notice that in either case $\tilde{x}_{i'r'} = 1$. Now consider 
the feasible solution $x \in \{0,1\}^{n\times k }$ for~\eqref{eq:qubononlinear} obtained from $\tilde{x}$ by letting $x = \X_{i'r'}(\tilde{x})$ (cf.,~\eqref{eq:mapx}). Notice that from~\eqref{eq:h0},~\eqref{eq:mapx}, and the fact that $\tilde{x}_{i'r'}=1$ one has that~\eqref{eq:h0diffeq} holds.
Also, from~\eqref{eq:h1n},~\eqref{eq:mapx}, and the fact that $\tilde{x}_{i'r'}=1$, it follows that $-H^n_1(x) = - H^n_1(\tilde{x}) + \sum_{(i',j\neq j') \in E} \tilde{x}_{jr'}+\tilde{x}_{j'r'}$. Thus,
\begin{equation}
\label{eq:h1ndiff}
-H^n_1(x) \ge - H^n_1(\tilde{x}) +\tilde{x}_{j'r'}.
\end{equation}
Further, from~\eqref{eq:h2n},~\eqref{eq:mapx}, and the fact that $\tilde{x}_{i'r'}=1$, it follows that 
\begin{equation}
\label{eq:h2ndiffeq}
-H^n_2(x) = - H^n_2(\tilde{x}) + \sum_{r \neq r' \in [k]}x_{i'r}. 
\end{equation}
Using~\eqref{eq:h0diffeq},~\eqref{eq:h1ndiff},~\eqref{eq:h2ndiffeq}, it follows that 
\begin{equation}
\label{eq:hndiffeq}
H^n_{c_1,c_2}(x) \ge H^n_{c_1,c_2}(\tilde{x}) -1+ c_1\tilde{x}_{j'r'} + c_2\sum_{r \neq r' \in [k]}x_{i'r}.\end{equation}
In case (1), we have that $\tilde{x}_{j'r'} = 1$. Thus, from~\eqref{eq:hndiffeq}, we have that $H^n_{c_1,c_2}(x) \ge H^n_{c_1,c_2}(\tilde{x}) -1+ c_1>H^n_{c_1,c_2}(\tilde{x}) = Q^n_{c_1,c_2}(k,G)$, which contradicts the optimality of $\tilde{x}$ for~\eqref{eq:qubononlinear}. 
Analogously, in case (2), we have that $\sum_{r \neq r' \in [k]} \tilde{x}_{j'r} \ge 1$. Thus, from~\eqref{eq:hndiffeq}, we have that
 $H^n_{c_1,c_2}(x) \ge H^n_{c_1,c_2}(\tilde{x}) -1+ c_2 > H^n_{c_1,c_2}(\tilde{x}) = Q^n_{c_1,c_2}(k,G)$, which contradicts the optimality of $\tilde{x}$ for~\eqref{eq:qubononlinear}. 

Therefore $\tilde{x}$ satisfies that there is no $(i',j') \in E$ and $r'\in [k]$ such that $\tilde{x}_{i'r'} + \tilde{x}_{j'r'} >1$, or $i' \in [n]$ and $r' \in [k]$ such that $\sum_{r \in [k]} \tilde{x}_{i'r} >1$. Therefore $\tilde{x}$ is a feasible solution of~\eqref{eq:MkCS}, which finishes the proof.
\end{proof}

Next, we show that the value of the penalty parameters $c_1, c_2$ in the definition of $Q^n_{c_1,c_2}(k,G)$ in~\eqref{eq:qubononlinear} can be further reduced to the values $c_1=c_2=1$ (indeed, more generally to $c_1 = 1$, $c_2 \ge 1$, or $c_1 \ge 1$, $c_2=1$), while still being able to obtain an optimal solution for the M$k$CS problem for $G(V,E)$ by solving the QUBO problem $Q^n_{1,1}(k,G)$. In this case, $Q^n_{1,1}(k,G)$ and $\alpha_k(G)$ are equivalent in terms of their optimal objective value, but not necessarily in terms of their optimal solutions. That is, the optimal solution~$\tilde{x}:= \argmax(Q^n_{1,1}(k,G))$ might not necessarily be a feasible solution for the M$k$CS problem, which hinders the possibility of constructing a M$k$CS set $H$ for $G(V,E)$. However, as we formally show in the next corollary, the $Q^n_{1,1}(k,G)$ optimal solution $\tilde{x}$
can be simply modified to obtain an optimal solution for the M$k$CS problem.


\begin{corollary}[unit-penalty nonlinear-based QUBO formulation of M$k$CS problem]
\label{cor:unitnonlinear}
Let $k\geq 1$ and a graph $G=(V,E)$ on $n$ vertices and $c_1, c_2 \ge 0$ be given, and let $\tilde{x} := \argmax\{Q^n_{c_1,c_2}(k,G)\}$ (recall~\eqref{eq:qubononlinear}). If $c_1 = 1$, $c_2 \ge 1$ or $c_1 \ge 1$, $c_2 = 1$, then
$Q^n_{c_1,c_2}(k,G) = \alpha_k(G)$. Furthermore,
$x' \in \argmax\{\alpha_k(G)\}$,
 where $x' \in \{0,1\}^{n \times k}$ is the output obtained when $k$, $G(V,E)$, $|V|$, $x = \tilde{x}$ is used as input in Algorithm~\ref{alg:AlgorithmA}.
\end{corollary}

\begin{proof}
The result follows from the proof of Theorem~\ref{thm:nonlinearQUBO} and Algorithm~\ref{alg:AlgorithmA}. More specifically, in the case $c_1 =1$, $c_2 \ge 1$, notice that Algorithm~\ref{alg:AlgorithmA}, step~\eqref{step:4}, is equivalent to applying the mapping $\X_{ir}(\cdot)$ (recall~\eqref{eq:mapx}) to the current solution $x$ in the Algorithm when $x_{ir} =x_{jr}=1$ for some $(i,j) \in E$, $r\in [k]$. Thus, it follows from~\eqref{eq:hndiffeq} that the value of $H^n_{c_1,c_2}=H^n_{1,c_2}(x)$ can only increase or stay equal 
after Algorithm~\ref{alg:AlgorithmA}, step~\eqref{step:4}. Similarly, in the case $c_1 \ge 1$, $c_2 = 1$, notice that Algorithm~\ref{alg:AlgorithmA}, step~\eqref{step:9}, is equivalent to applying the mapping $\X_{ir}(\cdot)$ (recall~\eqref{eq:mapx}) to the current solution $x$ in the Algorithm when $x_{ir} =1$, $\sum_{p \neq r \in[k]} x_{ip} \ge 1$ for some $i \in [n]$, $r \in [k]$. Thus, it follows from~\eqref{eq:hndiffeq} that the value $H^n_{c_1,c_2}(x) = H^n_{c_1,1}(x)$ can only increase of stay equal after Algorithm~\ref{alg:AlgorithmA}, step~\eqref{step:9}. Thus, in both cases, at the end of Algorithm~\ref{alg:AlgorithmA} one obtains a feasible solution $x'$ for the M$k$CS problem with objective $H^n_{c_1,c_2}(x') \ge H^n_{c_1,c_2}(\tilde{x}) = Q^n_{c_1,c_2}(k,G)$. Since $Q^n_{c_1,c_2}(k,G) \ge \alpha_k(G)$ (see beginning of proof of Theorem~\ref{thm:nonlinearQUBO}), it follows that $Q^n_{c_1,c_2}(k,G) = \alpha_k(G)$, and $x' \in \argmax\{\alpha_k(G))\}$.
\end{proof}

In light of Theorem~\ref{thm:nonlinearQUBO} and Corollary~\ref{cor:unitnonlinear}, it is natural to consider what happens if in the QUBO problem~\eqref{eq:qubononlinear} one considers penalty parameters $0< c_1, c_2 < 1$. 

\begin{proposition}
\label{prop:counter}
Let $k \ge 1$ and $c_1, c_2 > 0$ be given. If $c_1 < 1$ or $c_2 < 1$ and $k \ge 1$, then  
there exists a graph $G(V, E)$ such that $Q^n_{c_1,c_2}(k,G) > \alpha_k(G)$.
\end{proposition}

\begin{proof}
First, consider the case in which $0 < c_1 < 1$, and let $G(V, E)$ is a clique of $k+1$ vertices. Clearly $\alpha_k(G) = k$. Now, for all $i \in [k+1]$, $r \in [k]$, let 
\[
x_{ir} = \left \{ \begin{array}{ll}
                             1 & i=r, i \le k,\\
                             1 & i=k+1, r=k, \\
                             0 & \text{otherwise.}\\
                             \end{array} \right .
                             \]
Then, $Q^n_{c_1,c_2}(k,G) \ge H^n_{c_1,c_2}(x) = (k+1) - c_1 > k = \alpha_k(G)$. 
Now, consider the case in which $0 < c_2 < 1$, and let 
$G(V, E)$ be the graph on $k+1$ vertices obtained by taking a clique in $k+1$ vertices and adding a vertex $k+2$ and edge $(k+1,k+2)$. That is, $V = [k+2]$, and 
$E = \{(i,j): 1 \le i < j \le k+1\} \cup \{(k+1,k+2)\}$. Clearly $\alpha_k(G) = k+1$. Now let 
\[
x_{ir} = \left \{ \begin{array}{ll}
                             1 & i=r, i \le k\\
                             1 & i=k+2, r=k \\
                             1 & i=k+2, r=k-1 \\
                             0 & \text{otherwise}\\
                             \end{array} \right . , \text{ for all } i \in [k+2], r \in [k].
                             \]
Then, $Q^n_{c_1,c_2}(k,G) \ge H^n_{c_1,c_2}(x) = (k+2) - c_2 > k+1 = \alpha_k(G)$.                         
\end{proof}

\begin{remark}
\label{rem:summary}
Theorem~\ref{thm:nonlinearQUBO} together with Corollary~\ref{cor:unitnonlinear} and Proposition~\ref{prop:counter} fully characterize the QUBO problem~\eqref{eq:qubononlinear} as a means to obtain a QUBO reformulation of the M$k$CS problem. In short, for any $c_1, c_2 \ge 1$, solving the nonlinear-based QUBO problem~\eqref{eq:qubononlinear} is equivalent to solving the M$k$CS problem with the caveat that if either $c_1 = 1$ or $c_2 = 1$, the simple Algorithm~\ref{alg:AlgorithmA} might need to be applied to the optimal solution of~\eqref{eq:qubononlinear} in order to obtain an optimal solution for the M$k$CS problem.  On the other hand, if $0< c_1<1$ or $0<c_2<1$, solving the nonlinear-based QUBO problem~\eqref{eq:qubononlinear}  is not guaranteed to provide
the objective value or the solution to the M$k$CS problem.
\end{remark}

As illustrated in Section~\ref{sec:benchmark}, the full characterization provided in this section (see summary in Remark~\ref{rem:summary}) gives the freedom to fine tune the QUBO reformulation of the M$k$CS problem to make the best use of quantum tools in addressing the solution of this problem.

In finishing this section, recall that the M$k$CS problem is equivalent to the stable set problem when $k=1$. Thus the QUBO reformulation results~\citep[see, e.g.,][]{calude2017, pardalos1994maximum, nannicini2019performance, wocjan20032, boros2007local, harant2000some, abello2001finding, pajouh2013characterization} for the stable set problem of the form 
\begin{equation}
\label{eq:alphaQUBO}
\alpha(G)= \max \left \{\dsum_{i=1}^{n}x_i^2 - c_1\dsum_{(i,j) \in E} x_ix_j: x \in \{0,1\}^{n} \right \},
\end{equation}
for a given graph $G(V, E)$ on $n$ vertices and $c_1 \ge 1$ follow from Theorem~\ref{thm:nonlinearQUBO} and Corollary~\ref{cor:unitnonlinear}. In particular, Corollary~\ref{cor:unitnonlinear} implies results in which $c_1$ is set to one in~\eqref{eq:alphaQUBO}. However, Corollary~\ref{cor:unitnonlinear} brings up a fact that, to the best of our knowledge, has been ignored in the literature; namely, that when $c_1$ is set to one in~\eqref{eq:alphaQUBO}, the support of the optimal solution of~\eqref{eq:alphaQUBO} might not necessarily correspond to a stable set of the graph $G(V, E)$. However, an optimal solution for the stable set problem can be obtained from the optimal solution of~\eqref{eq:alphaQUBO} by applying Algorithm~\ref{alg:AlgorithmA} (see Corollary~\ref{cor:unitnonlinear}).

\subsection{Linear-based QUBO reformulation revisited}
\label{sec:linearagain}

After the results in Section~\ref{sec:nonlinear}, which provide a full characterization of the QUBO problem~\eqref{eq:qubononlinear} to reformulate the M$k$CS problem, it is natural to consider if a similar full characterization of the QUBO problem~\eqref{eq:qubolinear} can be obtained. Indeed, it is not difficult to see that analogous results (with analogous proofs that are not included in the interest of brevity), to Corollary~\ref{cor:unitnonlinear} and Proposition~\ref{prop:counter} can be obtained for the QUBO problem~\eqref{eq:qubolinear}.

\begin{corollary}[unit-penalty linear-based QUBO formulation of M$k$CS problem]
\label{cor:unitlinear}
Let $k\geq 1$ and a graph $G=(V,E)$ on $n$ vertices and $c_1, c_2 \ge 0$ be given, and let $\tilde{x} := \argmax_{x}\{Q^l_{c_1,c_2}(k,G)\}$ (recall~\eqref{eq:qubolinear}). If $c_1 = 1$, $c_2 \ge 1$ or $c_1 \ge 1$, $c_2 = 1$, then
$Q^l_{c_1,c_2}(k,G) = \alpha_k(G)$. Furthermore,
$x' \in \argmax\{\alpha_k(G)\}$,
 where $x' \in \{0,1\}^{n \times k}$ is the output obtained when $k$, $G(V,E)$, $|V|$, $x = \tilde{x}$ is used as input in Algorithm~\ref{alg:AlgorithmA}.
\end{corollary}

\begin{proposition}
\label{prop:counterlinear}
Let $k \ge 1$ and $c_1, c_2 > 0$ be given. If $c_1 < 1$ or $c_2 < 1$ and $k \ge 1$, then  
there exists a graph $G(V, E)$ such that $Q^l_{c_1,c_2}(k,G) > \alpha_k(G)$.
\end{proposition}

\section{Benchmarking}
\label{sec:benchmark}

To illustrate the benefits of the fully characterized QUBO reformulations presented in Section~\ref{sec:MkCS}, we next benchmark the linear-based (Section~\ref{sec:linear}) and nonlinear-based (Section~\ref{sec:nonlinear}) QUBO reformulations of the M$k$CS problem when solving them with a quantum annealer.
For this purpose, we present results pertaining the {\em minimum gap}~\citep[see, e.g.,][]{roland2002quantum}, related to the convergence rate of (an ideal) adiabatic quantum algorithm (AQC), {\em embedding}~\citep[see, e.g.,][]{vyskovcil2019embedding} into the available quantum annealing hardware, and {\em time-to-solution} (TTS)~\citep[see, e.g.,][]{ronnow2014defining} when performing the quantum annealing.

The embedding and TTS benchmarking results are obtained using D-Wave's quantum annealers (\url{https://www.dwavesys.com/}). 
Specifically, we report the different results obtained when using two different  D-Wave processors: 2000Q\textsuperscript{TM} and Advantage 1.1\textsuperscript{TM}.
The main difference between these two processors is the number of available qubits and their connectivity within the processor.
The 2000Q\textsuperscript{TM} processor has 2048 possible qubits (of which 2041 were available) connected in a {\em Chimera} connectivity graph, designed as a grid of $16 \times 16$ cells of $K_{4,4}$ bipartite graphs connected in a nearest-neighbor fashion by means of non-planar edges, where each qubit is connected to at most~6 neighbors~\citep{neven2009nips}.
The Advantage 1.1\textsuperscript{TM} processor counts with 5640 qubits (5510 available) following a {\em Pegasus} connectivity graph, defined as three layers of $16 \times 16$ cells of $K_{4,4}$ bipartite graphs with additional connections within and among the cells, providing an increased connectivity for each qubit to maximum 15 neighbors~\citep{boothby2020next}.

Each QUBO reformulation proposed here is thus used to solve the M$k$CS problem in a D-Wave quantum annealer. These numerical experiments are similar in nature to those carried out in~\citep{hua2020improved, calude2017, verma2020optimal, verma2020penalty} to compare different QUBO formulations of various COPT problems. 

To generate instances $G(V, E)$ for the numerical tests, given the number of nodes $|V| = n$, we generate instances with randomly defined edge set $E$, using Erd\H os-R\'enyi graphs $\G(n,p)$ with probabilities $p=0.25, 0.5, 0.75$ (i.e., with different levels of sparsity). In what follows, for the purpose of brevity, we will sometimes refer to the linear-based QUBO reformulation (Section~\ref{sec:linear}) as L-QUBO, and to the nonlinear-based QUBO reformulation (Section~\ref{sec:nonlinear}) as N-QUBO. All the classical computations are done using an Ubuntu Machine with processor Intel(R) Xeon(R) CPU E5-2630 v4 @ 2.20GHz, 94Gb of RAM, and 20 cores.

\subsection{Quantum Annealing}
Before presenting the benchmarking results, we provide a very brief high level discussion about the ideal version of the quantum annealing algorithm run by D-Wave's quantum annealer; that is, an AQC algorithm~\citep[following][]{cullimore2012relationship, DWAVEanneal}. For additional details, the reader is directed to \citep[][among many others]{farhi2000quantum, king2020performance,amin2008effect}. The fact that complex COPT problems (such as the M$k$CS problem considered here) can be reformulated as QUBO problems, means that finding the optimal solution of the problem can be regarded as finding or sampling low-energy states from an Ising spin model Hamiltonian~$\HH_f$ that is constructed from the QUBO formulation (i.e., using~\citep[][eq. (14)]{cullimore2012relationship}). For that purpose, a simple Hamiltonian~$\HH_i$ with an easily prepared ground (low-energy) state is prepared (i.e., using~\citep[][eq. (12)]{cullimore2012relationship}). Then, by constructing an adequate interpolation~\citep[see, e.g.,][eq. (1)]{cullimore2012relationship} between the $\HH_i$ and  $\HH_f$ Hamiltonians, the system can be set to slowly evolve (so that the {\em adiabatic} theorem~\citep[cf.,][]{farhi2000quantum} is satisfied) from the ground state of $\HH_i$ to a state that will yield the desired low-energy state of  $\HH_f$ with high probability. In particular, here we construct the interpolation:
\begin{equation}
\label{eq:interpolation}
\HH(s) = \frac{A(s)}{2}\HH_i + \frac{B(s)}{2}\HH_f,
\end{equation}
where $s \in [0,1]$ is the adimensional (or reduced) time~$s=\frac{t}{T}$ with~$t$ denoting time and~$T$ denoting the computation time, $A(s)$ is the tunneling energy curve, and $B(s)$ is the problem's Hamiltonian energy curve, for all~$s\in [0,1]$. The specifications of~$A(s)$ and~$B(s)$ for both the D-Wave 2000Q\textsuperscript{TM} and Advantage 1.1\textsuperscript{TM} processors can be found in D-Wave's documentation~\citep{DWAVEanneal}.

\subsection{Minimum Gap}
Now we consider the evolution of the AQC algorithm under~\eqref{eq:interpolation}. It is known that the {\em minimum gap}; that is, the minimum energy gap between the lowest two energy
levels during the evolution of the AQC algorithm determines the time of the
computation~\citep[see, e.g.,][]{farhi2001quantum,amin2012approximate}. Such energy levels correspond with the eigenvalues $E_m(s)$ of the eigenstates $m;s\rangle$ of the Hamiltonian $\HH(s)$, where $m \in \{0,1,\dots,2^{n}-1\}$ in an $n$ qubit system, and  are given by~\citep[see,][]{cullimore2012relationship}:
\begin{equation}
\HH(s)|m;s\rangle = E_m(s)|m;s\rangle, 
\end{equation}
with $E_0(s) \le E_1(s) \le \cdots \le E_{2^n-1}$. Thus, the minimum gap, denoted by $\Delta_{\min}$, is given by~\citep[see,][]{cullimore2012relationship}:
\begin{equation}
\label{eq:mingap}
\Delta_{\min} = \min_{s\in [0,1]}\{E_1(s)- E_0(s)\}.    
\end{equation}
The minimum gap $\Delta_{\min}$ provides a lower bound on the AQC computation time that is inversely proportional to $\Delta_{\min}^2$~\citep[see,][eq.~(9)]{cullimore2012relationship}; that is, the larger $\Delta_{\min}$, the faster the  AQC algorithm is expected to converge to the ground state of the Hamiltonian~$\HH_f$~\citep[see, e.g.,][]{amin2008effect}.

Here, we use the exact diagonalization of the instantaneous time-dependent Hamiltonian $\HH(s)$ to compute $\Delta_{\min}$. Although this methodology is well known to require a prohibitive amount of computation due to the need to diagonalize matrices of size $2^n \times 2^n$ for a system with $n$ qubits; it is suitable for the illustrative numerical tests performed here (for less computationally expensive ways to approximately compute $\Delta_{\min}$, we refer the readers to~\citep{amin2012approximate}). 

In particular, to obtain the minimum gap results presented next, we begin by calculating, for a number of finite values of $s \in [0,1]$ (see details below), the Hamiltonian~$\HH(s)$ in~\eqref{eq:interpolation} for particular instances of the L-QUBO (resp. N-QUBO) formulation of the M$k$CS problem. Then, we verify which eigenvalues of $\HH(s)$ correspond to different states in the annealing process. Two states are considered different if at the end of the annealing, their energy difference is more than a given $\varepsilon$. Here, we use $\varepsilon = 1$ GHz. Finally, we approximately compute the minimum difference along the annealing scaled time $s$ of the two smallest eigenvalues corresponding to different states (the ground state and the first excited state). The approximation comes from the fact that the minimum in~\eqref{eq:mingap} is computed over a finite set of values of $s \in [0,1]$. Specifically, given the monotonic behavior of $A(s)$ and $B(s)$, we observe that $\Delta_{\min}$ is attained at a value $s \in [0,1]$ that is close to the value of $s \in [0,1]$ in which the maximum of the minimum eigenvalue is attained. Moreover, for the instances considered here, 
the minimum eigenvalue is concave on~$s \in [0,1]$, allowing us to more efficiently sample the domain $[0,1]$ in search for the value of $\Delta_{\min}$. Namely,
we first consider a coarse discretization of the domain~$[0,1]$ 
(taking only $10$ equally spaced elements of the set) and then determine the three points whose middle point would be larger than both its neighbors.
This interval contains the values of $s \in [0,1]$ in which both the maximum value of the smallest eigenvalue, and $\Delta_{\min}$ is attained.
We sample this interval by computing $10$ points between each of these points (including them) to refine the approximation to the value $s^*$ in which the minimization in~\eqref{eq:mingap} is attained.
This procedure only requires~$44$ computations of the eigenvalues of the Hamiltonian~$\HH(s)$.

\subsubsection{Minimum Gap Results}
\label{sec:mingap}
Next, we compare the minimum gap ($\Delta_{\min}$) resulting when using the L-QUBO reformulation and the N-QUBO reformulations for the M$k$CS problem with varying penalty parameters (i.e., $c_1$, $c_2$) on small instances of the M$k$CS problem.

\begin{figure}[!htb]
    \centering
    \begin{minipage}[t]{.45\textwidth}
        \centering
        \includegraphics[width=0.9\linewidth]{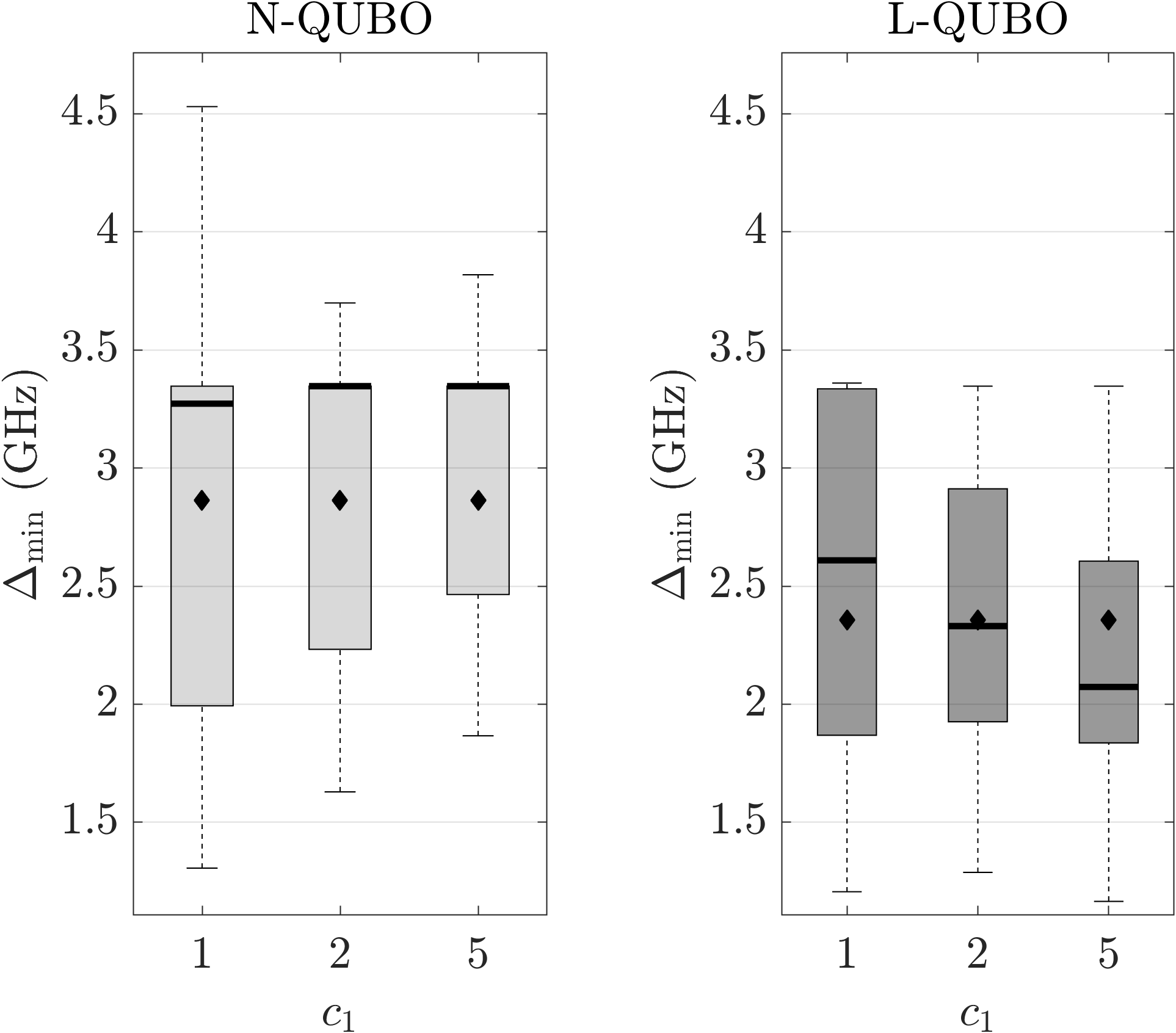}
        \caption{$k=1$, $\G(5,0.25)$. \label{fig:mingap1_25}}
    \end{minipage}
    \begin{minipage}[t]{0.45\textwidth}
        \centering
        \includegraphics[width=0.9\linewidth]{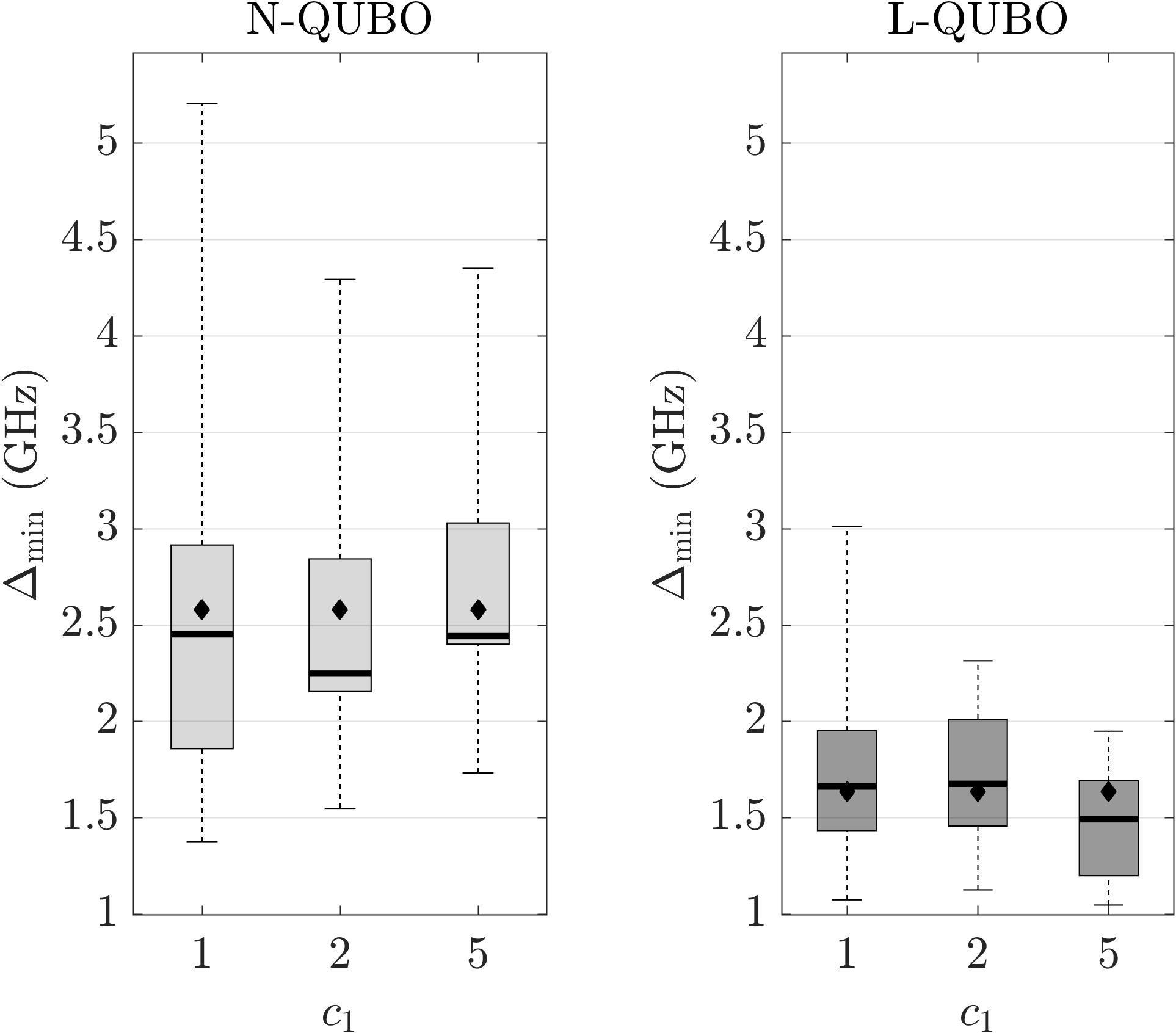}
        \caption{$k=1$, $\G(5,0.75)$. \label{fig:mingap1_75}}
    \end{minipage}
\end{figure}

In particular, Figures~\ref{fig:mingap1_25} and~\ref{fig:mingap1_75}
compare the $\Delta_{\min}$ obtained from the L-QUBO~\eqref{eq:qubolinear} and N-QUBO~\eqref{eq:qubononlinear} for instances of the M$k$CS problem in which $k=1$, where the underlying graphs are randomly selected $\G(5,0.25)$ and $\G(5,0.75)$ graphs. These bar plots, as well as the remaining ones in this section, provide information about the distribution of $\Delta_{\min}$. Specifically, each bar is obtained by computing $\Delta_{\min}$ on $100$ randomly generated graphs $\G(5,p)$, $p \in \{0.25, 0.50, 0.75\}$, for different values of the  penalization parameters. The tick line, represents the median of $\Delta_{\min}$, the black diamond represents $\overline{\Delta}_{\min}$, the average of $\Delta_{\min}$, the bar encompasses the values within the 25\% and 75\% quantiles of $\Delta_{\min}$'s distribution, and the dotted interval encompasses the values within the $0$\% and $100$\% of $\Delta_{\min}$'s distribution.

From Figures~\ref{fig:mingap1_25} and~\ref{fig:mingap1_75}, it follows that the N-QUBO results in higher $\overline{\Delta}_{\min}$ than the L-QUBO; therefore, in theory, the N-QUBO Hamiltonian should converge faster to a low energy state than the L-QUBO Hamiltonian. We can state this more formally by performing a simple hypothesis test. Let $\overline{\Delta}_{\min}^a(k,p,c_1,c_2)$ (resp. $\mu_{{\Delta}_{\min}^a}(k,p,c_1,c_2)$) be the average (resp. mean) of the minimum gap for the $a$-QUBO formulation of the M$k$CS instance from $\G(5,p)$ graphs, and penalty parameters $c_1$, $c_2$. Then, consider the hypothesis test:
\begin{equation}
\label{eq:hypo1}
\begin{array}{llll}
H_o: & \mu_{{\Delta}_{\min}^{\rm L}}(1,p,1,\cdot)  \ge 
\mu_{{\Delta}_{\min}^{\rm N}}(1,p,1,\cdot) - \delta \overline{\Delta}_{\min}^{\rm N}(1,p,1,\cdot)\\
H_a: & \mu_{{\Delta}_{\min}^{\rm L}}(1,p,1,\cdot)  < 
\mu_{{\Delta}_{\min}^{\rm N}}(1,p,1,\cdot) - \delta \overline{\Delta}_{\min}^{\rm N}(1,p,1,\cdot),\\
\end{array}
\end{equation}
where $\delta \in [0,100\%]$. That is, in~\eqref{eq:hypo1} we are statistically comparing the left-most bars of the N-QUBO subplot and the L-QUBO subplot of Figures~\ref{fig:mingap1_25} and~\ref{fig:mingap1_75} (for brevity, the results for the case $p=0.50$ have not bee plotted), under the null hypothesis that the L-QUBO provides a higher mean $\Delta_{\min}$. Then, for any $p \in \{0.25, 0.50, 0.75\}$ one gets that the null hypothesis $H_o$ in~\eqref{eq:hypo1} can be rejected with 95\% confidence for values of $\delta$ up to $2\%$. Thus, loosely speaking, the N-QUBO results in values of $\Delta_{\min}$ that on average are 2\% higher than the ones obtained by the L-QUBO, when using penalty parameter $c_1 = 1$.

One advantage of having the full characterization of the penalty constants, for which the QUBO formulations~\eqref{eq:qubolinear} and~\eqref{eq:qubononlinear} become reformulations of the M$k$CS problem, is that we can investigate what are the trade-offs of increasing such penalty values from their minimum ones. Intuitively, one might expect that $\Delta_{\min}$ increases (faster convergence) as the values of the penalty parameters $c_1, c_2$ increase. This reasoning stem from the fact that higher penalty parameters increase the suboptimality of infeasible solutions of the original problem in its associated QUBO reformulation. However, from Figures~\ref{fig:mingap1_25} (left) and~\ref{fig:mingap1_75} (left) it follows that $\overline{\Delta}_{\min}$ remains fairly unchanged as the penalty parameter $c_1$ increases. More formally consider a similar hypothesis test to the one considered in~\eqref{eq:hypo1}.
\begin{equation}
\label{eq:hypo2}
\begin{array}{llll}
H_o: & \mu_{{\Delta}_{\min}^{\rm N}}(1,p,c_1,\cdot)  \ge 
\mu_{{\Delta}_{\min}^{\rm N}}(1,p,1,\cdot) + \delta \overline{\Delta}_{\min}^{\rm N}(1,p,1,\cdot)\\
H_a: & \mu_{{\Delta}_{\min}^{\rm N}}(1,p,c_1,\cdot)  < 
\mu_{{\Delta}_{\min}^{\rm N}}(1,p,1,\cdot) + \delta \overline{\Delta}_{\min}^{\rm N}(1,p,1,\cdot).\\
\end{array}
\end{equation}
Then, for any $p \in \{0.25, 0.50, 0.75)\}, c_1 \in \{2,5\}$, one gets that the null hypothesis $H_o$ in~\eqref{eq:hypo1} can be rejected with 95\% confidence for values of $\delta$ less than $1\%$. Thus, loosely speaking, the N-QUBO with penalty parameter $c_1=1$ results in values of $\Delta_{\min}$ that on average are not 1\% lower than the ones obtained by the N-QUBO with higher penalty parameters $c_1 \in \{2,5\}$.

It is clearly interesting to investigate how the characteristics above look when considering higher values of $k$. Due to the complexity of the exact diagonalization procedure used to compute $\Delta_{\min}$, we limit to the study of the case $k=2$ for the N-QUBO reformulation, considering penalty parameters constants $c_1 \in \{1,2,5\}$, and $c_2 \in \{1,2,5\}$.

\begin{figure}[!htb]
    \centering
    \begin{minipage}[t]{.45\textwidth}
        \centering
        \includegraphics[width=0.9\linewidth]{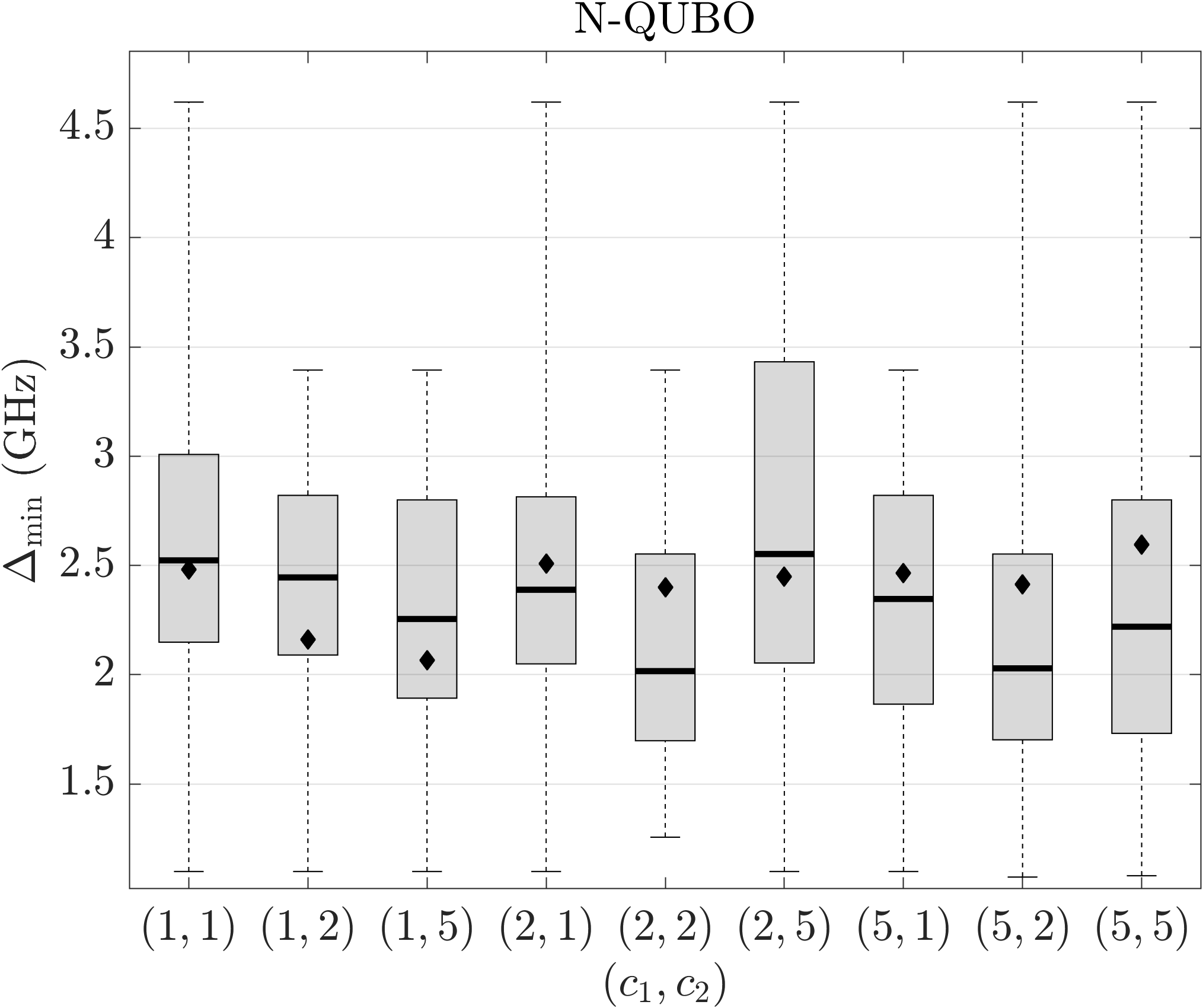}
        \caption{$k=2$, $\G(5,0.25)$. \label{fig:mingap2_25}}
    \end{minipage}
    \begin{minipage}[t]{0.45\textwidth}
        \centering
        \includegraphics[width=0.9\linewidth]{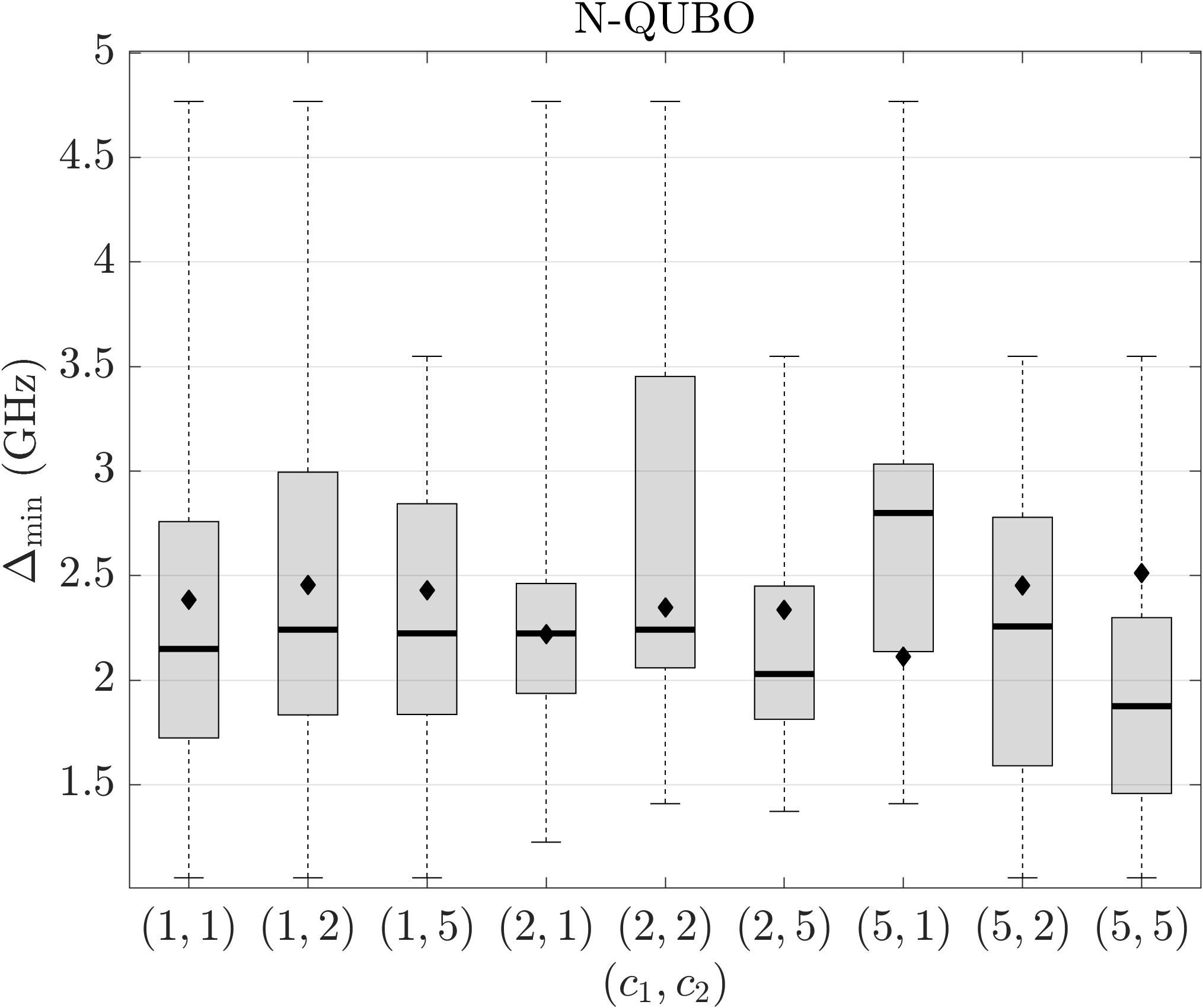}
        \caption{$k=2$, $\G(5,0.75)$. \label{fig:mingap2_75}}
    \end{minipage}
\end{figure}

Much like in the case when $k=1$, from Figures~\ref{fig:mingap2_25} and~\ref{fig:mingap2_75} it follows that $\overline{\Delta}_{\min}$ barely increases as the penalty parameters $c_1, c_2$ increase, in comparison to the value of $\overline{\Delta}_{\min}$ when the penalty parameters are set to $c_1 = c_2 = 1$. Statistically, things are a bit different. Formally, consider a similar hypothesis test to the one considered in~\eqref{eq:hypo2}.
\begin{equation}
\label{eq:hypo3}
\begin{array}{llll}
H_o: & \mu_{{\Delta}_{\min}^{\rm N}}(2,p,c_1,c_2)  \ge 
\mu_{{\Delta}_{\min}^{\rm N}}(2,p,1,1) + \delta \overline{\Delta}_{\min}^{\rm N}(2,p,1,1)\\
H_a: & \mu_{{\Delta}_{\min}^{\rm N}}(2,p,c_1,c_2)  < 
\mu_{{\Delta}_{\min}^{\rm N}}(2,p,1,1) + \delta \overline{\Delta}_{\min}^{\rm N}(2,p,1,1).\\
\end{array}
\end{equation}
Table~\ref{tab:deltas} shows the minimum value of $\delta$ in~\eqref{eq:hypo3} for which the null hypothesis $H_o$ in~\eqref{eq:hypo2} can be rejected with a 95\% confidence level. Loosely speaking, the value of $\delta$ indicates the percentage by which the value of $\overline{\Delta}_{\min}(2,p,1,1)$ must be higher in order to reject the hypothesis that larger penalty parameters (i.e., larger than $c_1 =c_2 = 1$) result in a larger mean value of $\Delta_{\min}$.

\begin{table}[!htb]
{\footnotesize
\begin{center}
\begin{tabular}{cccccc}
\toprule
                     &\multicolumn{3}{c}{$\delta$}                &         \\
                     \cmidrule{2-4}
$(c_1, c_2)$ & $p=0.25$ & $p=0.50$ & $p=0.75$ & $H_o$\\
\midrule
$(1,2)$ & -5\% & 8\% &  12\% & Reject\\
$(1,5)$ & -5\% & 5\% &  11\% & Reject\\
$(2,1)$ & 1\% & 4\% &  1\% &Reject\\
$(2,2)$ & 1\% & 12\% &  7\% &Reject\\
$(2,5)$ & 1\% & 13\% &  6\% &Reject\\
$(5,1)$ & 1\% &2\% &  0\% & Reject\\
$(5,2)$ & 4\% & 15\% &  11\% &Reject\\
$(5,5)$ & 11\% &22\% &  13\% & Reject\\
\bottomrule
\end{tabular}
\caption{Hypothesis test~\eqref{eq:hypo3} for different parameters with 95\% confidence. \label{tab:deltas}}
\end{center}
}
\end{table}

Overall, there is a recognizable pattern in Table~\ref{tab:deltas}. Namely, it is clear that the sparser the underlying graph (i.e., lower probability~$p$) used to construct the instance of the M$k$CS problem, the smaller the effect of increasing penalty parameters is on increasing the mean of $\Delta_{\min}$ (i.e., accelerating convergence of an AQC algorithm). This is intuitively expected, given that sparsity in the underlying graph results in lower number of penalty terms in the N-QUBO~\eqref{eq:qubononlinear}. In particular, notice that a significant increase in the mean of~$\Delta_{\min}$, for sparse underlying graphs (i.e., $p=0.25$), only arises when the penalty parameters are increased from $c_1=c_2 = 1$ to $c_1=c_2 = 5$. In contrast, for non-sparse underlying graphs (i.e., $p=0.75$), increases in the penalty constants above $c_1=c_2 = 1$ bring increases in the mean of~$\Delta_{\min}$ of about 10\% in most cases. In Section~\ref{sec:TTS}, we will analyze how these increases in the mean of~$\Delta_{\min}$ affect the convergence to a solution in D-Wave's quantum annealing devices. Before doing this analysis, we first consider the differences in terms of embedding requirements between the N-QUBO and L-QUBO formulation.


\subsection{Embedding}
\label{sec:embed}

Current NISQ devices have a low number of qubits available with restricted connectivity.
Given this, the number of qubits required to {\em embed}~\citep[see, e.g.,][]{vyskovcil2019embedding} a QUBO reformulation of a given COPT problem in a quantum device is a very important benchmark to compare the benefits of different QUBO reformulations of a COPT problem~\citep[see, e.g.,][]{hua2020improved, fowler2017improved, verma2020optimal, verma2020penalty}. Next, to benchmark the N-QUBO~\eqref{eq:qubononlinear} versus the L-QUBO~\eqref{eq:qubolinear} reformulation of the M$k$CS problem, we 
use the number of qubits needed to embed the QUBO reformulation in both a 2048 qubits
{\em Chimera} connectivity graph (for D-Wave's 2000Q\textsuperscript{TM} processor), and a 5640 qubits {\em Pegasus} connectivity graph (for D-Wave's Advantage 1.1\textsuperscript{TM} processor). 
For this purpose, we use D-Wave's embedding algorithm~\citep[see, e.g.,][for a discussion of different embedding algorithms]{bernal2020integer, zbinden2020embedding}.

\begin{figure}[!htb]
    \centering
    \begin{minipage}[t]{.45\textwidth}
        \centering
        \includegraphics[width=\linewidth]{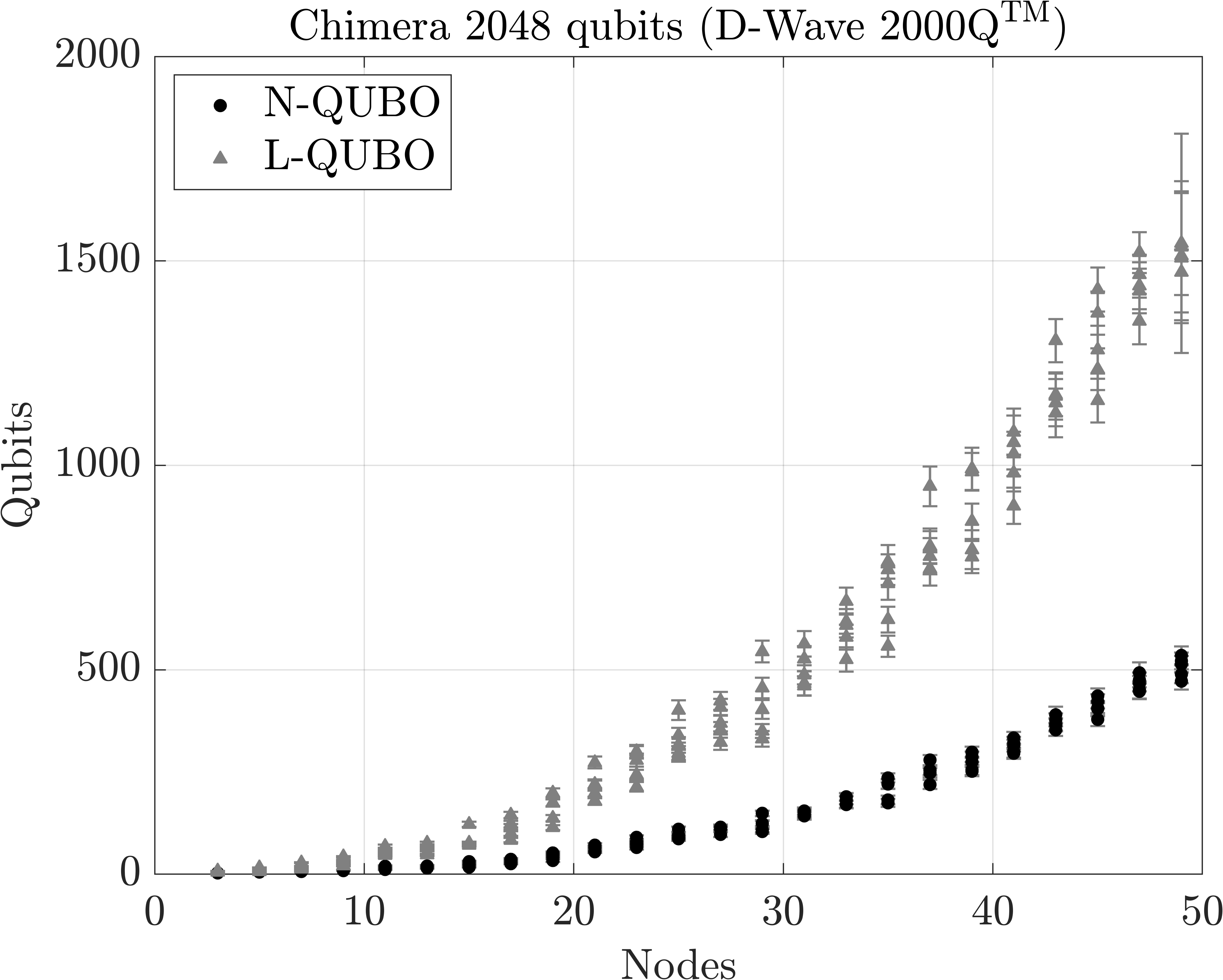}
        \caption{Embeding: $k=1$, $\G(n,0.25)$. \label{fig:embed1chi_25}}
    \end{minipage}
    \begin{minipage}[t]{0.45\textwidth}
        \centering
        \includegraphics[width=\linewidth]{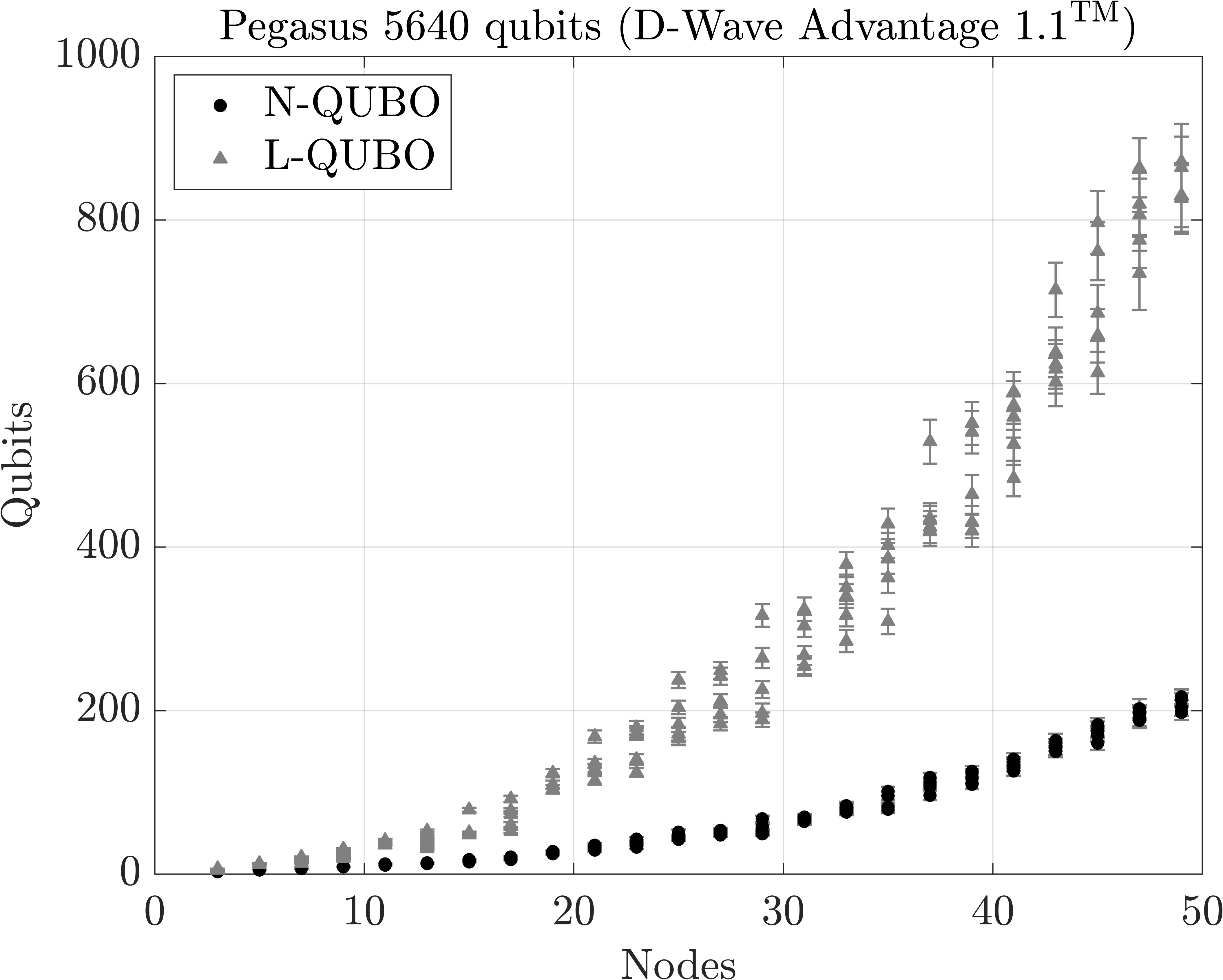}
        \caption{Embedding: $k=1$, $\G(n,0.25)$. \label{fig:embed1peg_25}}
    \end{minipage}
\end{figure}

Note that for a graph $G(V,E)$ with $n$ nodes, the number of binary variables required to formulate the L-QUBO for the associated M$k$CS problem is $k(n+ |E| +1)$, while $kn$ binary variables are needed to formulate the associated N-QUBO. Not surprisingly, the N-QUBO would require less number of qubits than the L-QUBO when these QUBOs are embedded into D-Wave's quantum annealers. The following results show how the increased number of binary variables required by the L-QUBO affects the difference between the qubits required to embed both QUBO formulations.

\begin{figure}[!htb]
    \centering
    \begin{minipage}[t]{.45\textwidth}
        \centering
        \includegraphics[width=\linewidth]{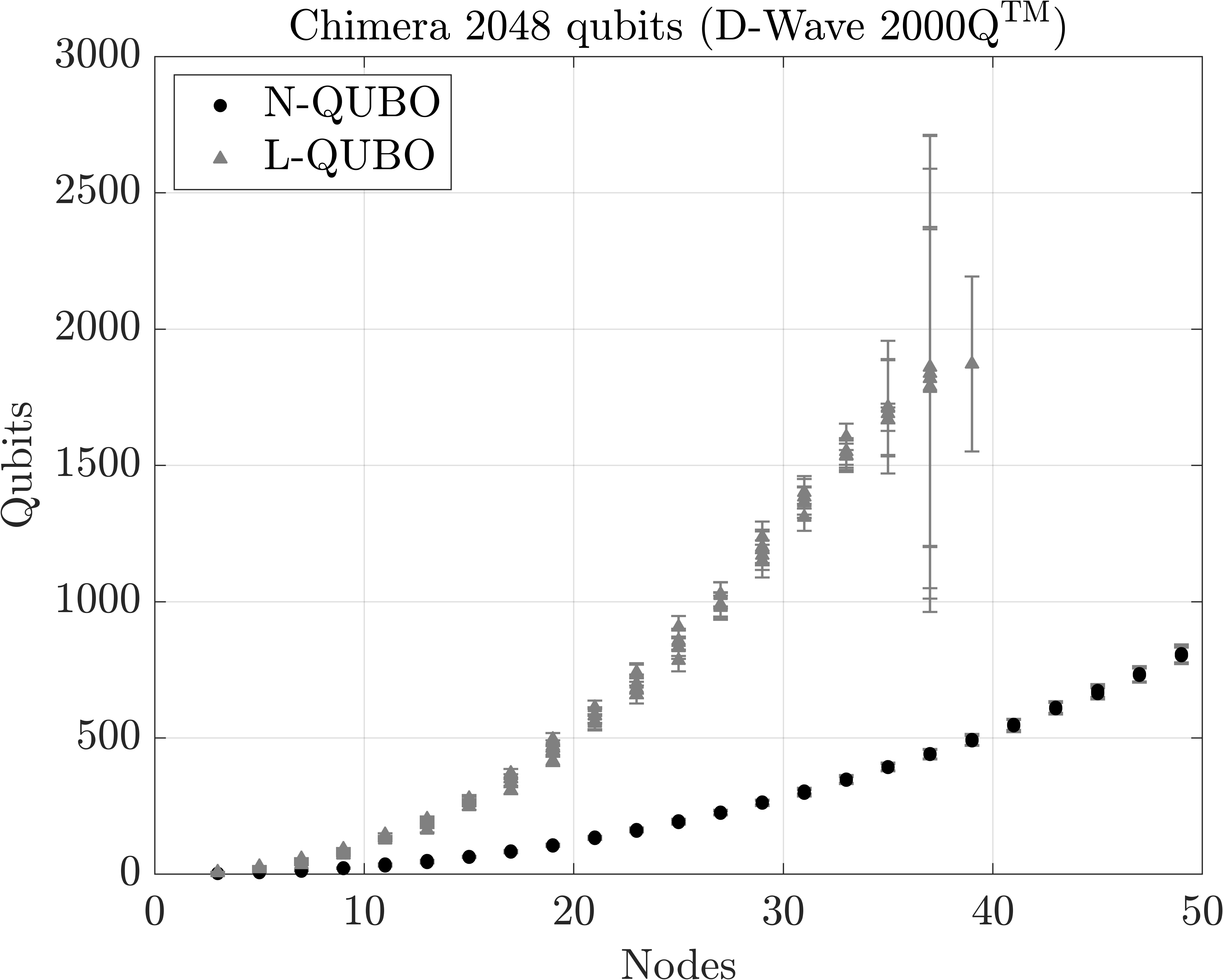}
        \caption{Embeding: $k=1$, $\G(n,0.75)$. \label{fig:embed1chi_75}}
    \end{minipage}
    \begin{minipage}[t]{0.45\textwidth}
        \centering
        \includegraphics[width=\linewidth]{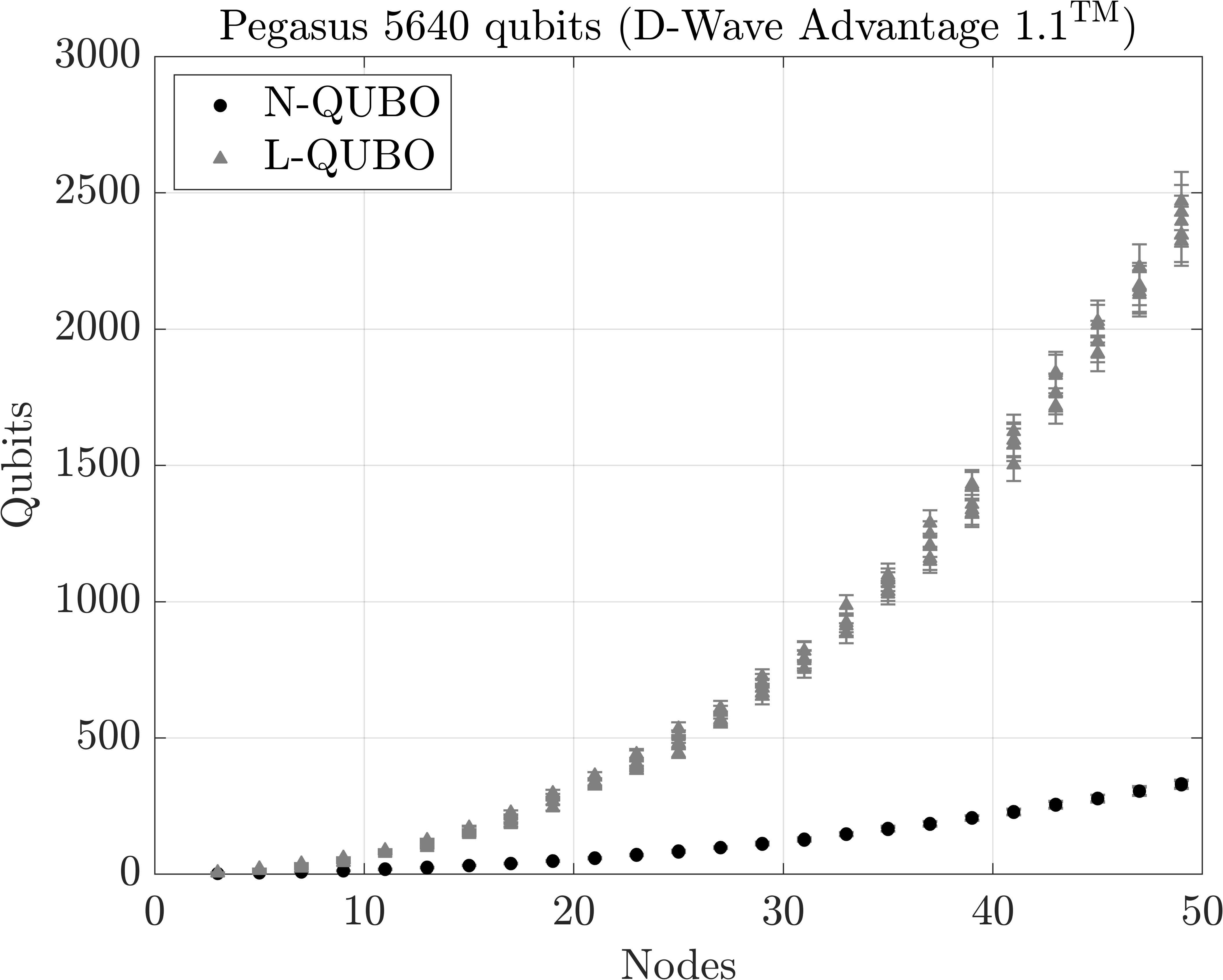}
        \caption{Embedding: $k=1$, $\G(n,0.75)$. \label{fig:embed1peg_75}}
    \end{minipage}
\end{figure}

In Figures~\ref{fig:embed1chi_25}-\ref{fig:embed5peg_25}, the number of average qubits required by both the L-QUBO and the N-QUBO formulations are plotted for values of $k\in \{1,2,5\}$, graphs $\G(n,p)$ for values of $n \in [5,50]$, $p \in \{0.25, 0.50, 0.75\}$, and D-Wave's  2000Q\textsuperscript{TM} and Advantage 1.1\textsuperscript{TM} processors. The average is computed over five (5) random graphs $\G(n,p)$ generated for each combination of $n, p$ values, as well as ten (10) runs of D-Wave's embedding algorithm. The bars plotted with each point in the graph represent the values within one standard deviation of the average value. 

\begin{figure}[!htb]
    \centering
    \begin{minipage}[t]{.45\textwidth}
        \centering
        \includegraphics[width=\linewidth]{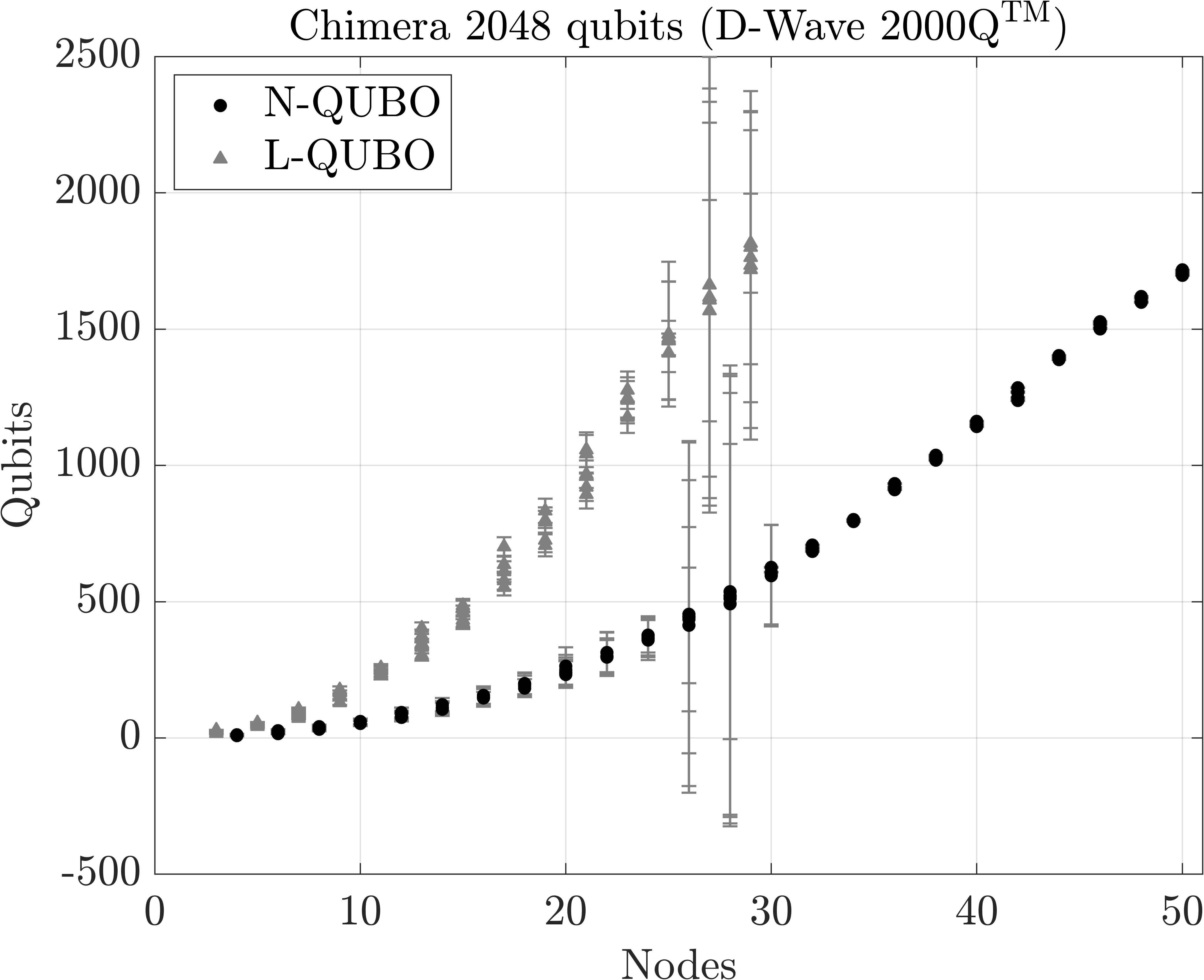}
        \caption{Embeding: $k=2$, $\G(n,0.50)$. \label{fig:embed2chi_50}}
    \end{minipage}
    \begin{minipage}[t]{0.45\textwidth}
        \centering
        \includegraphics[width=\linewidth]{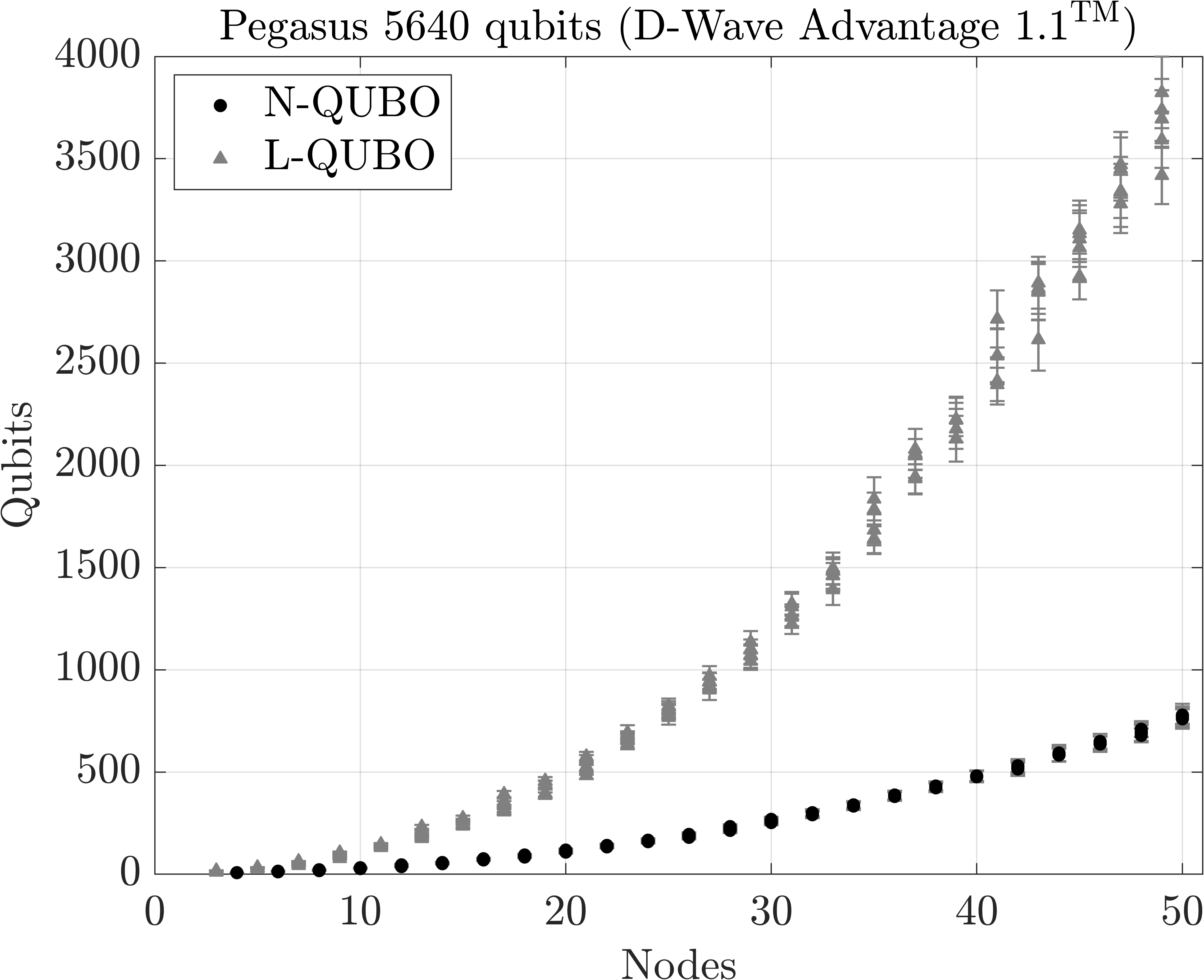}
        \caption{Embedding: $k=2$, $\G(n,0.50)$. \label{fig:embed2peg_50}}
    \end{minipage}
\end{figure}

From all Figures~\ref{fig:embed1chi_25}-\ref{fig:embed5peg_25}, it is clear that in terms of embedding requirements, the N-QUBO formulation is substantially better than the L-QUBO formulation. This is true not only in terms of the average qubits required to embed each QUBO, but the volatility of the number of qubits required to embed each QUBO. In Figure~\ref{fig:embed1chi_25}, in which sparse graphs (i.e., $p=0.25$) are used for the case $k=1$, both QUBO formulations can be embedded, for graphs with up to $n=50$, in D-Wave's 2000Q\textsuperscript{TM} processor. However, in Figure~\ref{fig:embed1chi_75}, where dense graphs (i.e., $p=0.75$) are considered, now the L-QUBO can be embedded only for graphs with up to $n=40$. From Figures~\ref{fig:embed2chi_50} and~\ref{fig:embed5chi_25} it is clear that as $k$ and $p$ increase, this trend of being able to embed larger problems in terms of number of nodes $n$ continues to be evidenced even more. Even using the more {\em powerful} Advantage 1.1\textsuperscript{TM} processor, Figure~\ref{fig:embed5peg_25}  shows that for sparse graphs (i.e., $p=0.25$) and $k=5$, the L-QUBO can only be embedded for graphs with up to $n=40$, while it seems that the N-QUBO can be embedded for graphs with up to $n=80$ (i.e., the double number of nodes).

\begin{figure}[!htb]
    \centering
    \begin{minipage}[t]{.45\textwidth}
        \centering
        \includegraphics[width=\linewidth]{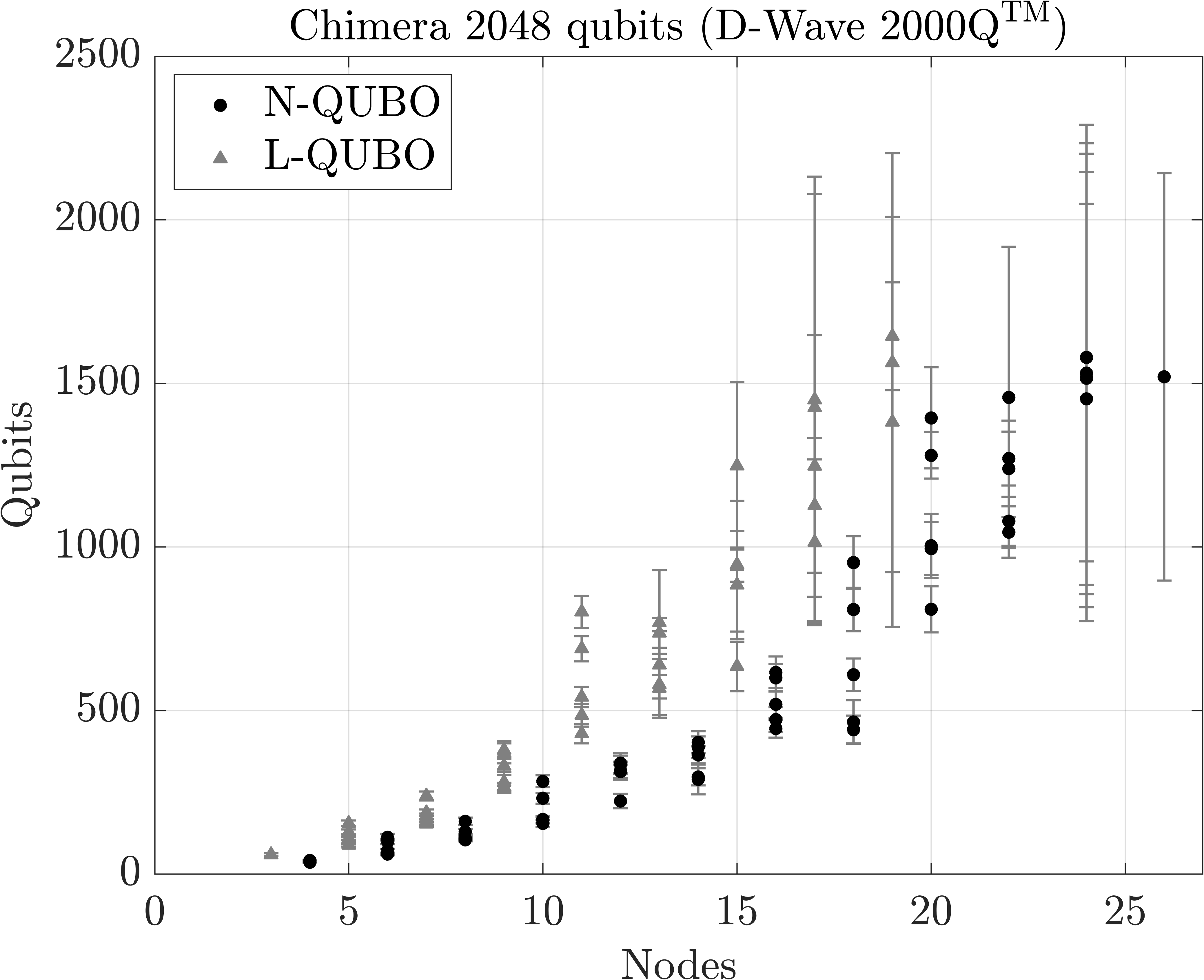}
        \caption{Embeding: $k=5$, $\G(n,0.25)$. \label{fig:embed5chi_25}}
    \end{minipage}
    \begin{minipage}[t]{0.45\textwidth}
        \centering
        \includegraphics[width=\linewidth]{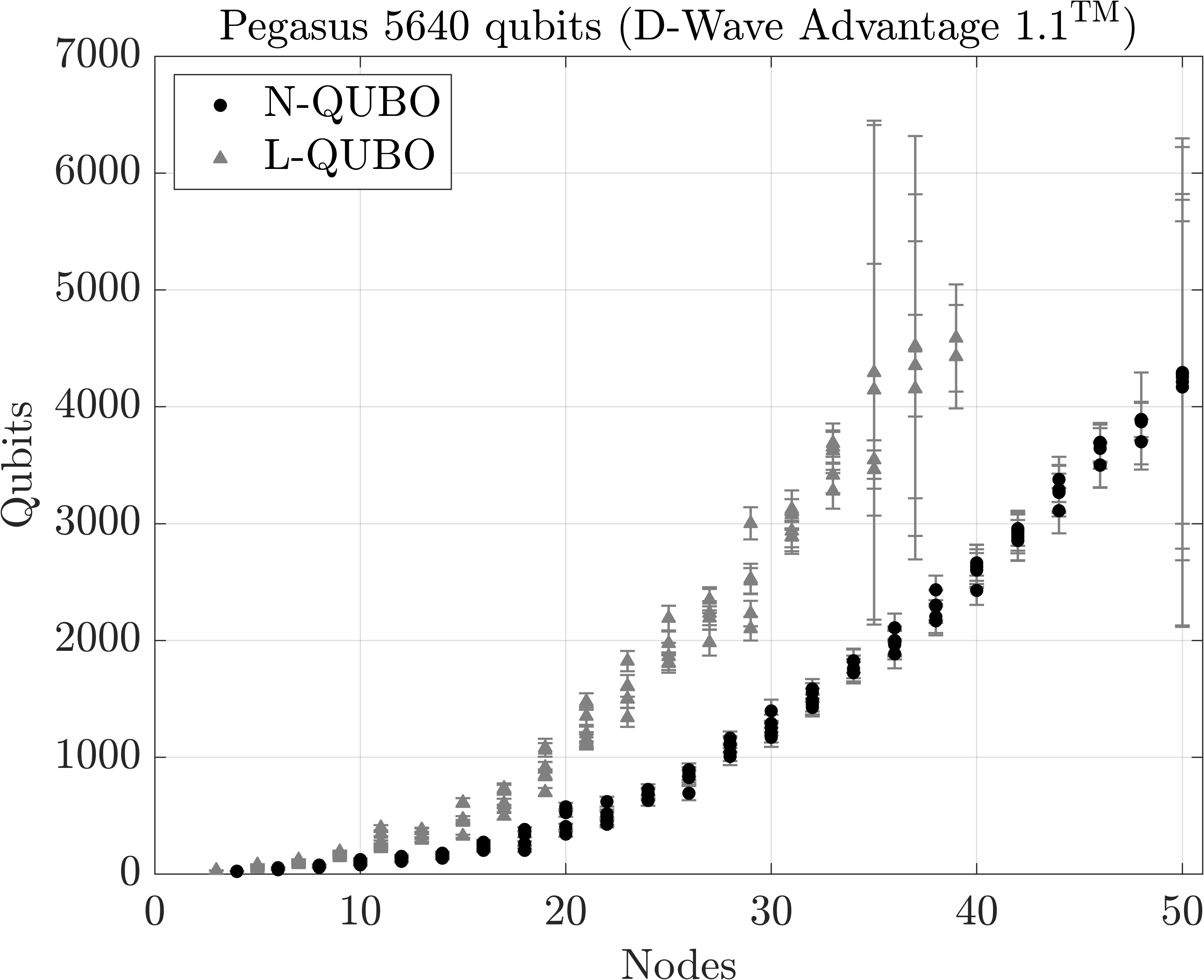}
        \caption{Embedding: $k=5$, $\G(n,0.25)$. \label{fig:embed5peg_25}}
    \end{minipage}
\end{figure}

By pairwise comparing Figures~\ref{fig:embed1chi_25}, \ref{fig:embed1chi_75}, \ref{fig:embed2chi_50}, \ref{fig:embed5chi_25}, versus Figures~\ref{fig:embed1peg_25}, \ref{fig:embed1peg_75}, \ref{fig:embed2peg_50}, \ref{fig:embed5peg_25}, one can see the advantages of the Advantage 1.1\textsuperscript{TM} processor versus the 2000Q\textsuperscript{TM} processor. The effect of having a larger number of qubits clearly means that larger instances of the L-QUBO can be embedded in the Advantage 1.1\textsuperscript{TM} processor than in the 2000Q\textsuperscript{TM} processor. Also evident are the effects of the improved connectivity between qubits in the Advantage 1.1\textsuperscript{TM} processor. Namely, it is clear that the Advantage 1.1\textsuperscript{TM} processor is able to embed QUBO problems using a substantially lower number of qubits than in the 2000Q\textsuperscript{TM} processor. For example, from Figures~\ref{fig:embed1chi_75} and~\ref{fig:embed1peg_75}, it takes about 800 qubits to embed the N-QUBO of graphs with 50 nodes in the 2000Q\textsuperscript{TM} processor, while it takes about half the number of qubits (about 400) to embed the N-QUBO of graphs with 50 nodes in the Advantage 1.1\textsuperscript{TM} processor. The pairwise comparison between Figures~\ref{fig:embed1chi_25}, \ref{fig:embed1chi_75}, \ref{fig:embed2chi_50}, \ref{fig:embed5chi_25}, and Figures~\ref{fig:embed1peg_25}, \ref{fig:embed1peg_75}, \ref{fig:embed2peg_50}, \ref{fig:embed5peg_25}, also shows how the added connectivity in the Advantage 1.1\textsuperscript{TM} processor clearly lowers the volatility of the number of qubits required to embed a QUBO using D-Wave's embedding algorithm.

\subsection{Time-To-Solution}
\label{sec:TTS}
We now finish our numerical tests by comparing (mirroring some of the tests in Sections~\ref{sec:mingap} and~\ref{sec:embed}) the {\em time-to-solution} (TTS) required, by both D-Wave's quantum annealer processors, when using the N-QUBO and L-QUBO reformulation of the M$k$CS, with penalty parameters $c_1 \in \{1,2,5\}, c_e \in \{1,2,5\}$, for values $k =2$, on random graphs $\G(n,p)$ for $n \in [5,50]$, $p \in \{0.25, 0.75\}$ (similar results were obtained for $k\in \{1,5\}$, and $p=0.50$ but are not presented for brevity). Besides benchmarking the N-QUBO versus the L-QUBO reformulation of the M$k$CS, these tests will be used to analyze the effect of the value of penalization constants in the TTS, and the effect in the TTS of using the more powerful Advantage 1.1\textsuperscript{TM} processor.

\begin{figure}[!htb]
    \centering
    \begin{minipage}[t]{.45\textwidth}
        \centering
        \includegraphics[width=\linewidth]{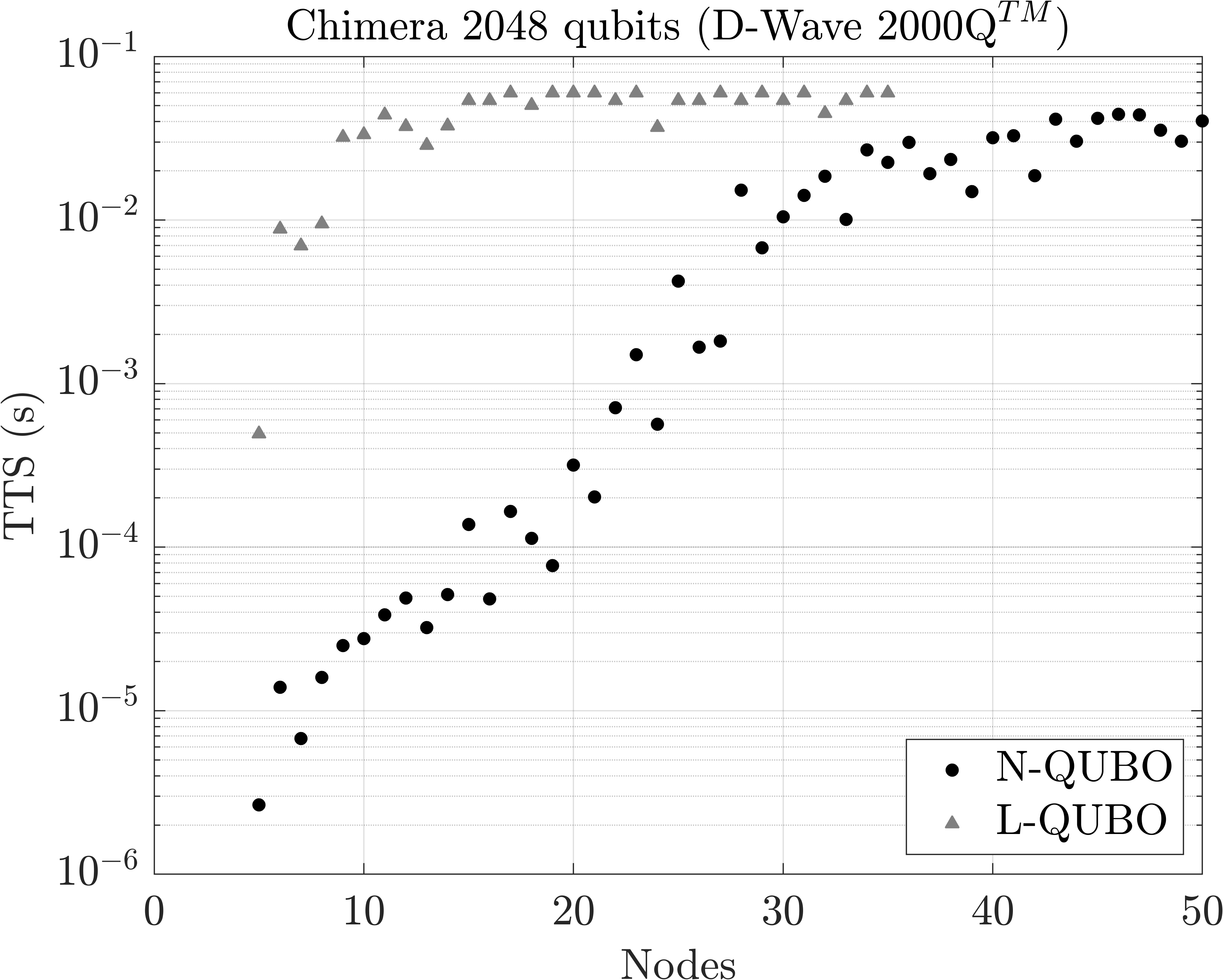}
        \caption{$k=2$, $\G(n,0.25)$, $c_1=c_2=1$. \label{fig:TTS_Chi_k2_p25_c1(1)_c2(1)}}
    \end{minipage}
    \begin{minipage}[t]{0.45\textwidth}
        \centering
        \includegraphics[width=\linewidth]{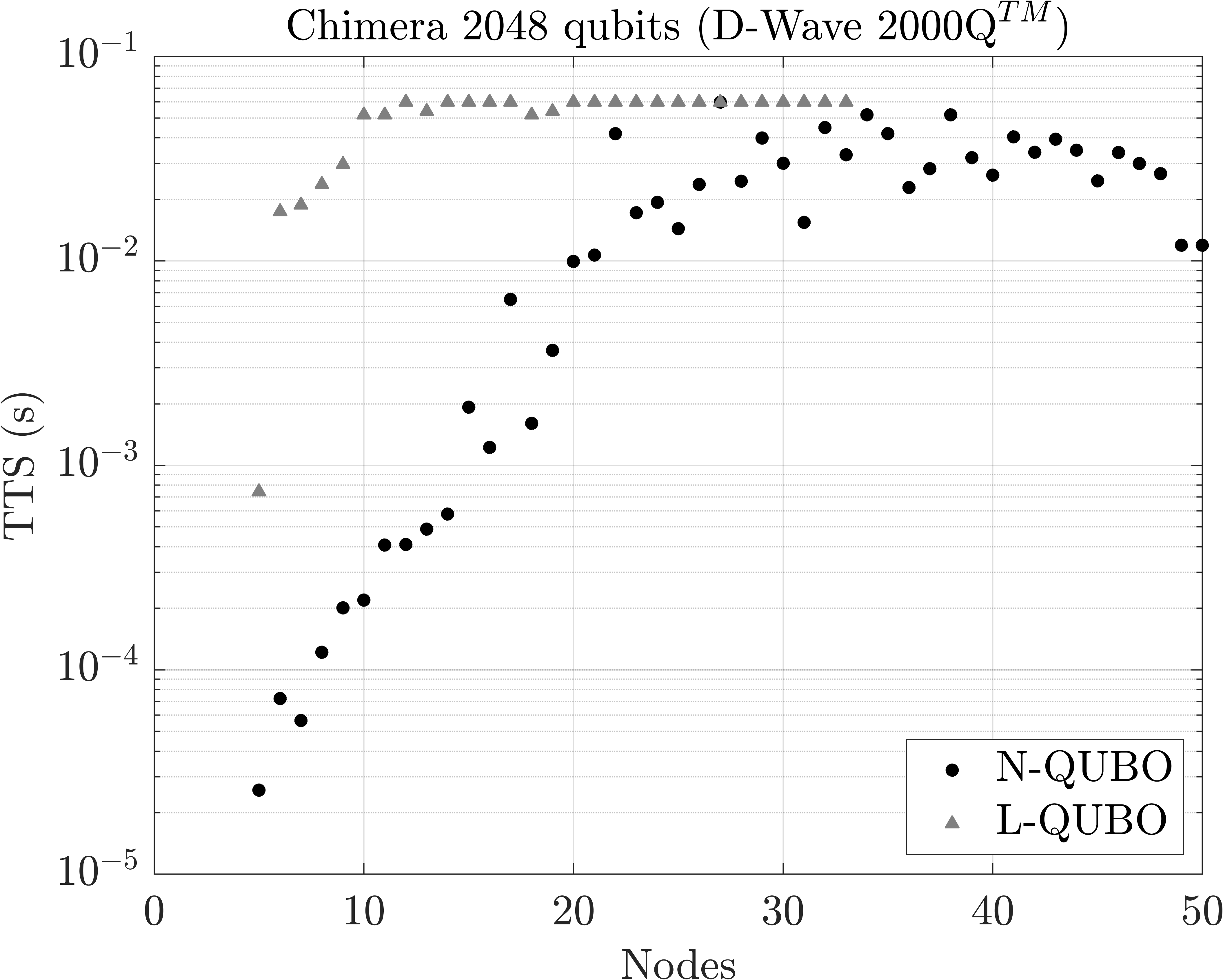}
        \caption{$k=2$, $\G(n,0.25)$, $c_5=c_2=5$. \label{fig:TTS_Chi_k2_p25_c1(5)_c2(5)}}
    \end{minipage}
\end{figure}

TTS is a common benchmark used to evaluate the performance of quantum annealers, and is defined as the expected time required to find a ground state of the desired Hamiltonian with a level of confidence $\alpha$, which is set to 95\% in our tests. Formally~\citep[see, e.g.,][]{ronnow2014defining},
\begin{equation}
\label{eq:TTS}
{\rm TTS} = t_{\rm run} \frac{\ln(1-\alpha)}{\ln(1-p)},
\end{equation}
where $t_{\rm run}$, fixed to 20$\mu$s in our tests, is the running time elapsed in a single run of the quantum annealer, and $p$ is the probability of finding the ground state of the desired Hamiltonian. In our tests, $p$ is estimated by running the quantum annealer 1000 times.

\begin{figure}[!htb]
    \centering
    \begin{minipage}[t]{.45\textwidth}
        \centering
        \includegraphics[width=\linewidth]{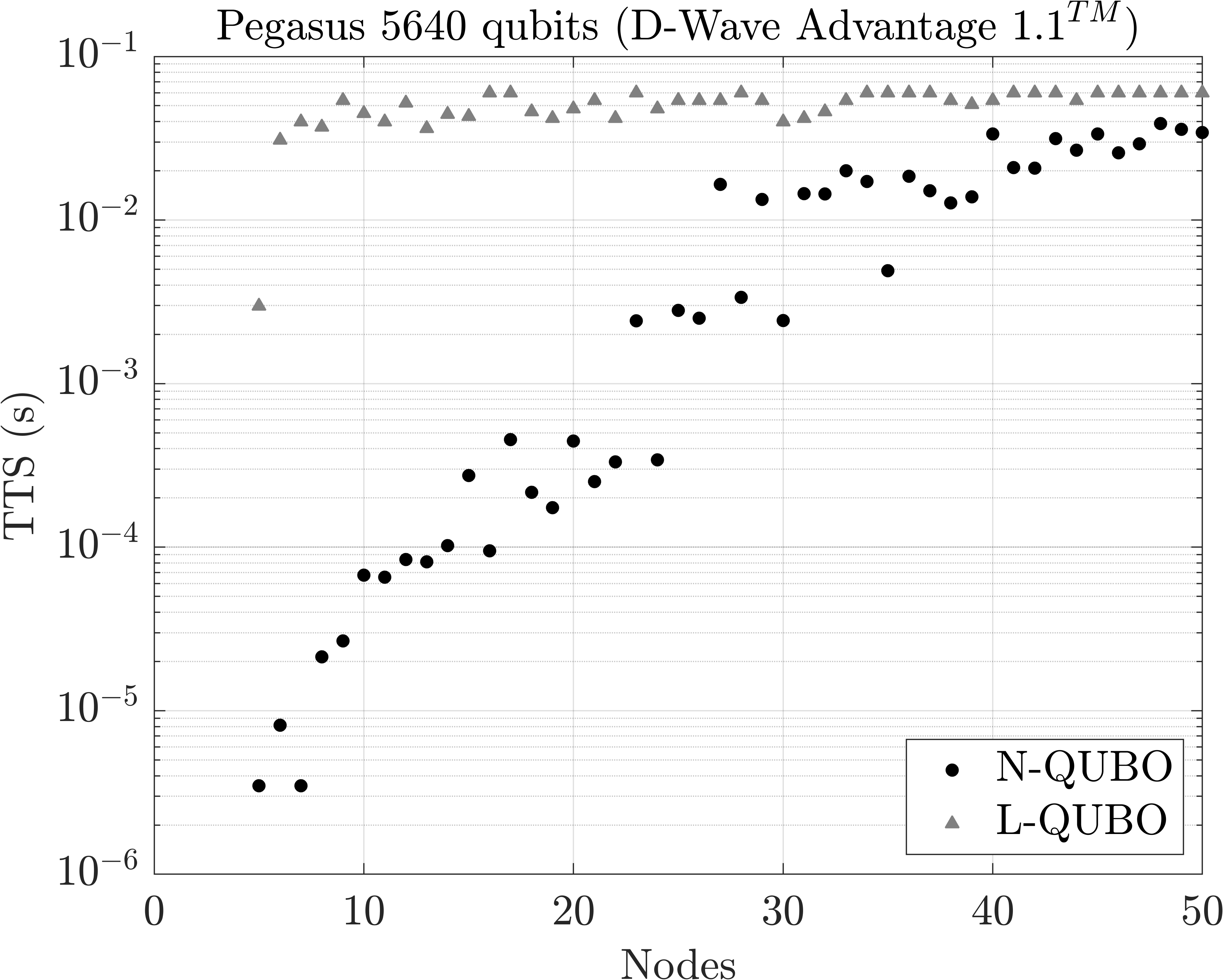}
        \caption{$k=2$, $\G(n,0.25)$, $c_1=c_2=1$. \label{fig:TTS_Peg_k2_p25_c1(1)_c2(1)}}
    \end{minipage}
    \begin{minipage}[t]{0.45\textwidth}
        \centering
        \includegraphics[width=\linewidth]{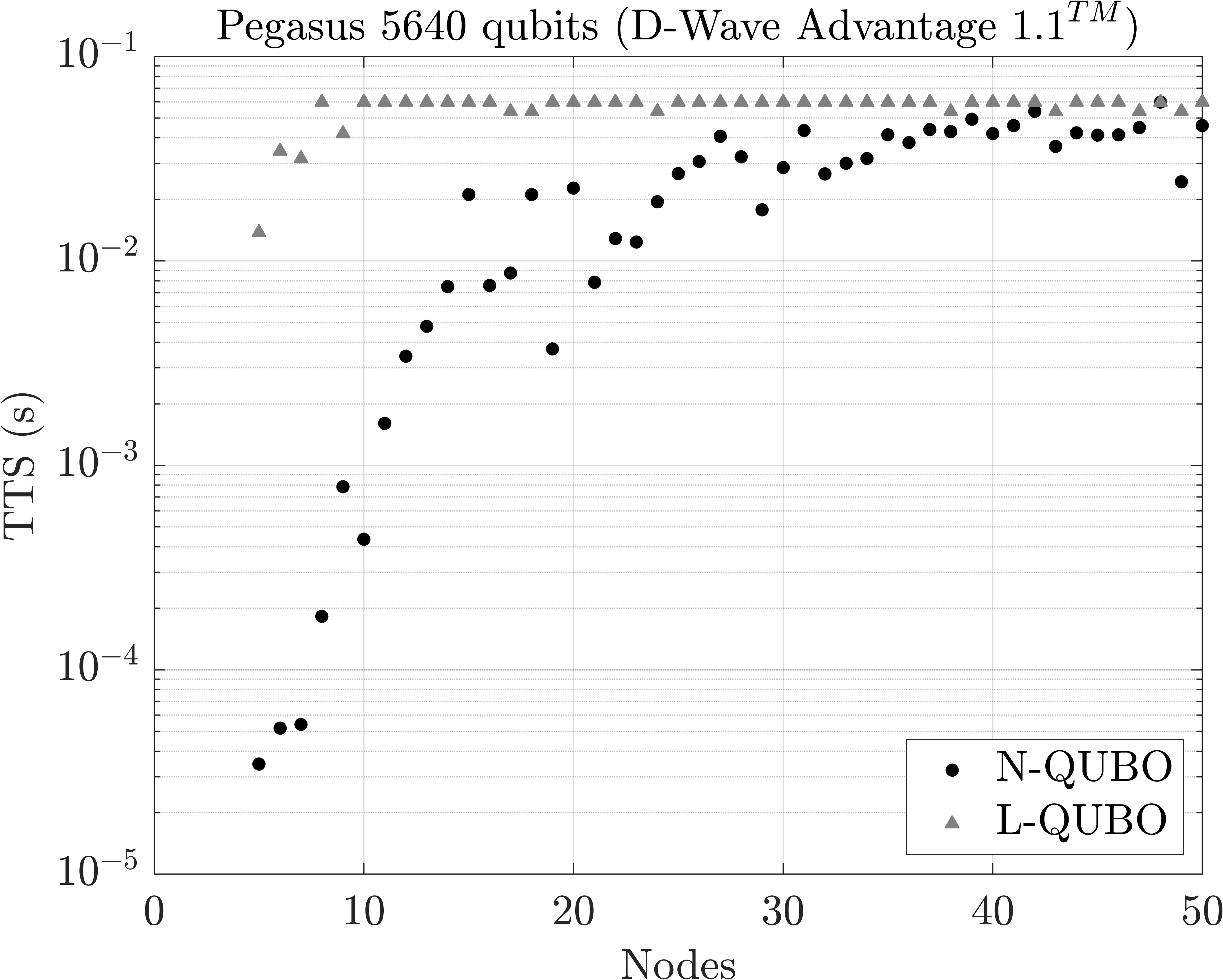}
        \caption{$k=2$, $\G(n,0.25)$, $c_5=c_2=5$. \label{fig:TTS_Peg_k2_p25_c1(5)_c2(5)}}
    \end{minipage}
\end{figure}

From Figures~\ref{fig:TTS_Chi_k2_p25_c1(1)_c2(1)}-\ref{fig:TTS_Peg_k2_p75_c1(5)_c2(5)}, it is clear that regardless of the quantum annealing processor, penalty parameters, or sparsity of the graph, the N-QUBO reformulation of the M$k$CS problem performs substantially better than the  associated L-QUBO reformulation in terms of TTS. For example, note that from Figure~\ref{fig:TTS_Peg_k2_p25_c1(1)_c2(1)} it follows that for sparse graphs (i.e., $p=0.25$), the N-QUBO results in TTS values that are between three (3) orders of magnitude faster, for small graphs, and one (1) order of magnitude faster, for larger graphs, than the associate TTS values for the L-QUBO. 

\begin{figure}[!htb]
    \centering
    \begin{minipage}[t]{.45\textwidth}
        \centering
        \includegraphics[width=\linewidth]{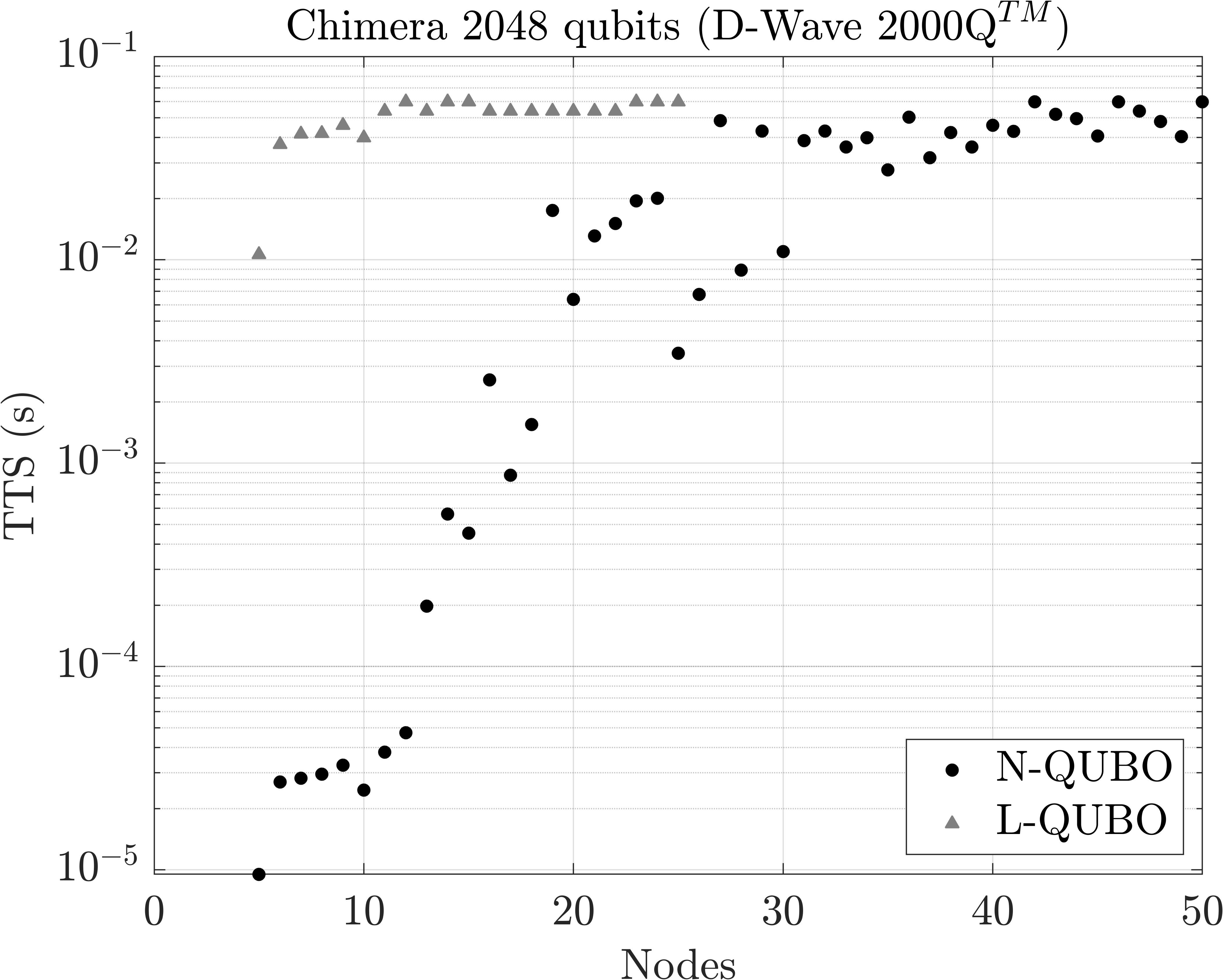}
        \caption{$k=2$, $\G(n,0.75)$, $c_1=c_2=1$. \label{fig:TTS_Chi_k2_p75_c1(1)_c2(1)}}
    \end{minipage}
    \begin{minipage}[t]{0.45\textwidth}
        \centering
        \includegraphics[width=\linewidth]{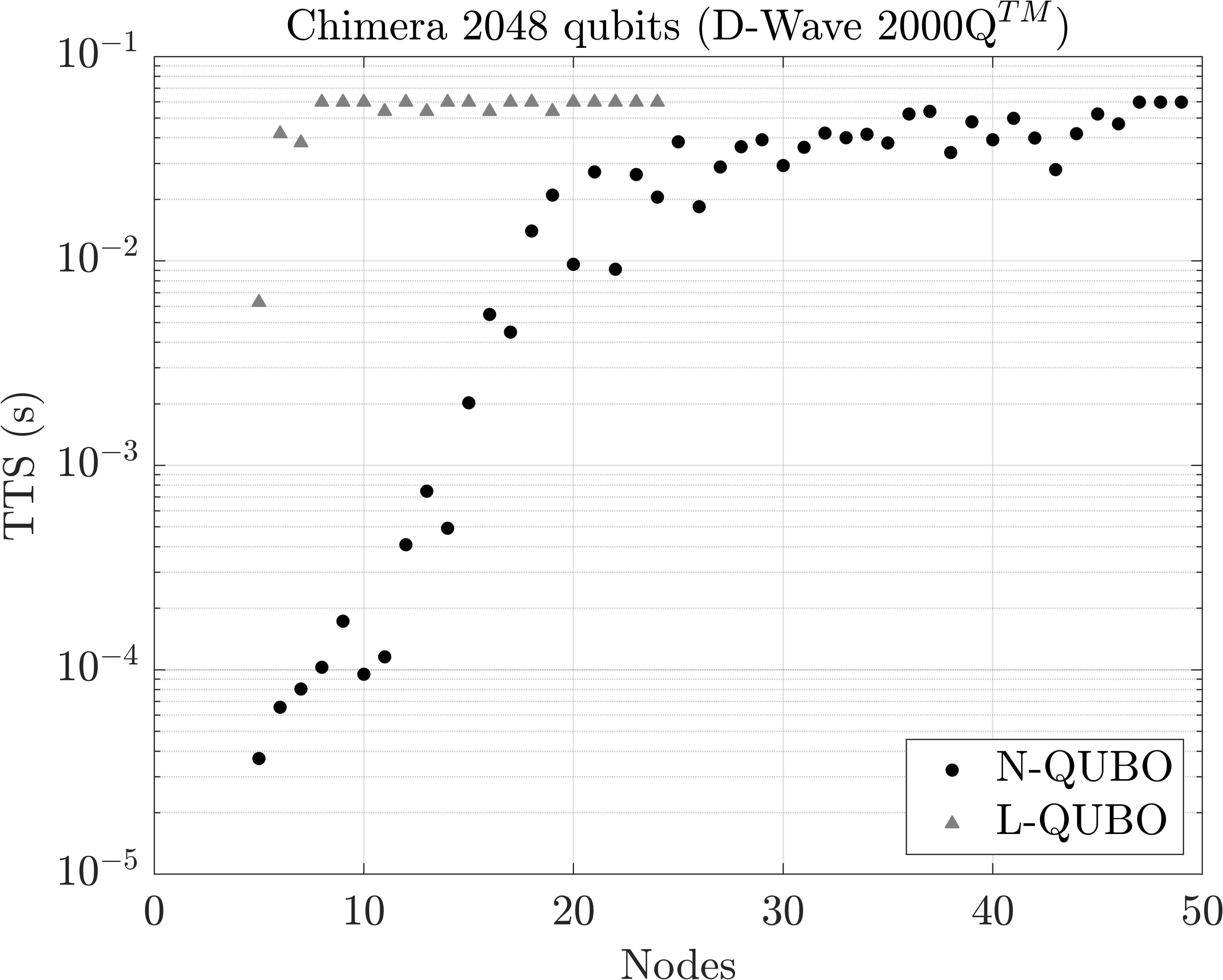}
        \caption{$k=2$, $\G(n,0.75)$, $c_5=c_2=5$. \label{fig:TTS_Chi_k2_p75_c1(5)_c2(5)}}
    \end{minipage}
\end{figure}

From Figure~\ref{fig:TTS_Chi_k2_p75_c1(1)_c2(1)} it is clear that when considering non-sparse graphs (i.e., $p=0.75$) in the 2000Q\textsuperscript{TM} processor, the advantages of the N-QUBO over the L-QUBO in terms of TTS only increase. In particular, notice that while a $t_{\rm run}$ of $20\mu$s is enough to find the optimal solution with some small probability for instances of the M$k$CS problem with underlying graphs of up to $n=50$ nodes. In contrast, once the number of nodes of the underlying graph goes beyond $n=25$, with the L-QUBO the quantum annealer is unable to find any optimal solution in all 1000 runs of $20\mu$s.

\begin{figure}[!htb]
    \centering
    \begin{minipage}[t]{.45\textwidth}
        \centering
        \includegraphics[width=\linewidth]{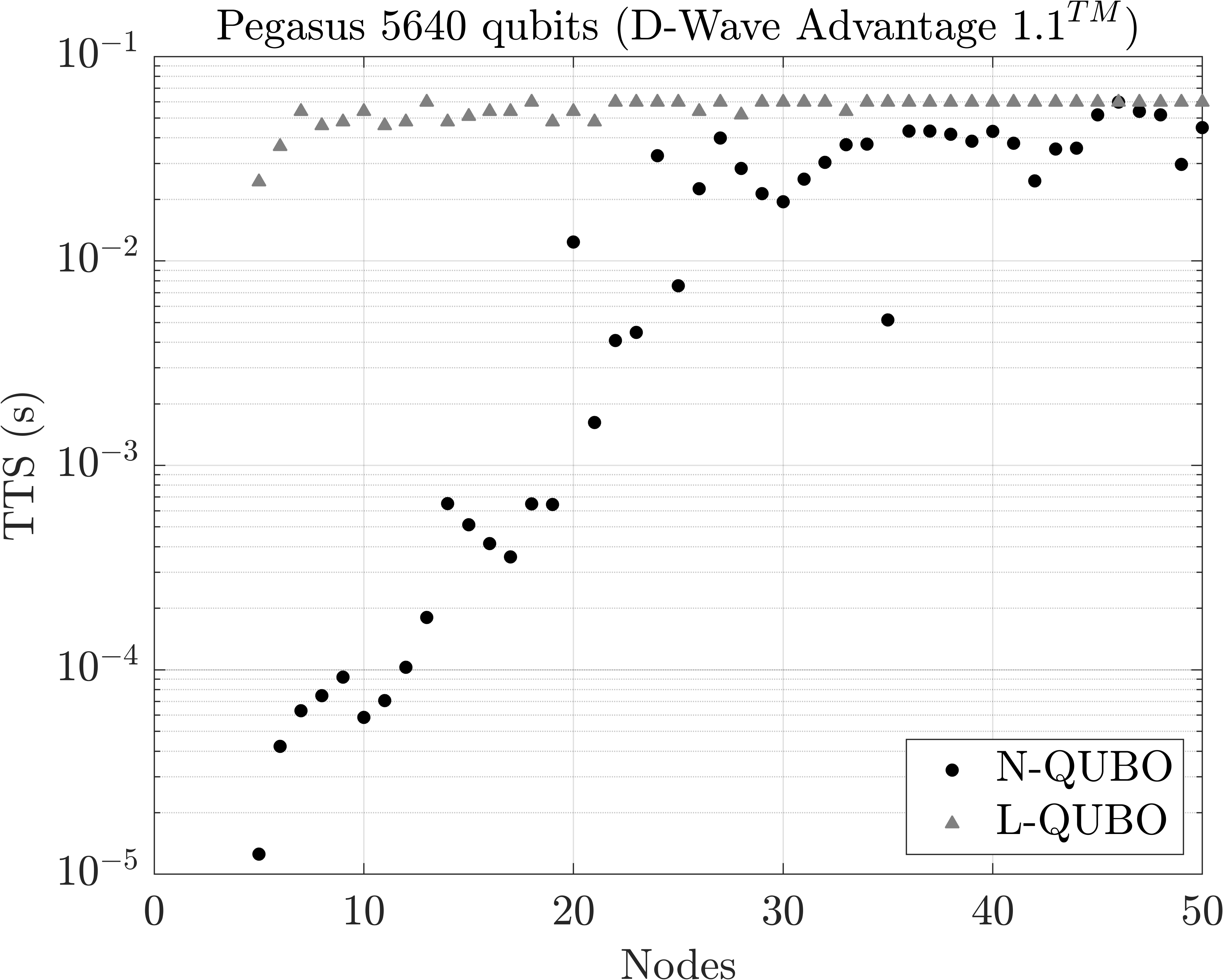}
        \caption{$k=2$, $\G(n,0.75)$, $c_1=c_2=1$. \label{fig:TTS_Peg_k2_p75_c1(1)_c2(1)}}
    \end{minipage}
    \begin{minipage}[t]{0.45\textwidth}
        \centering
        \includegraphics[width=\linewidth]{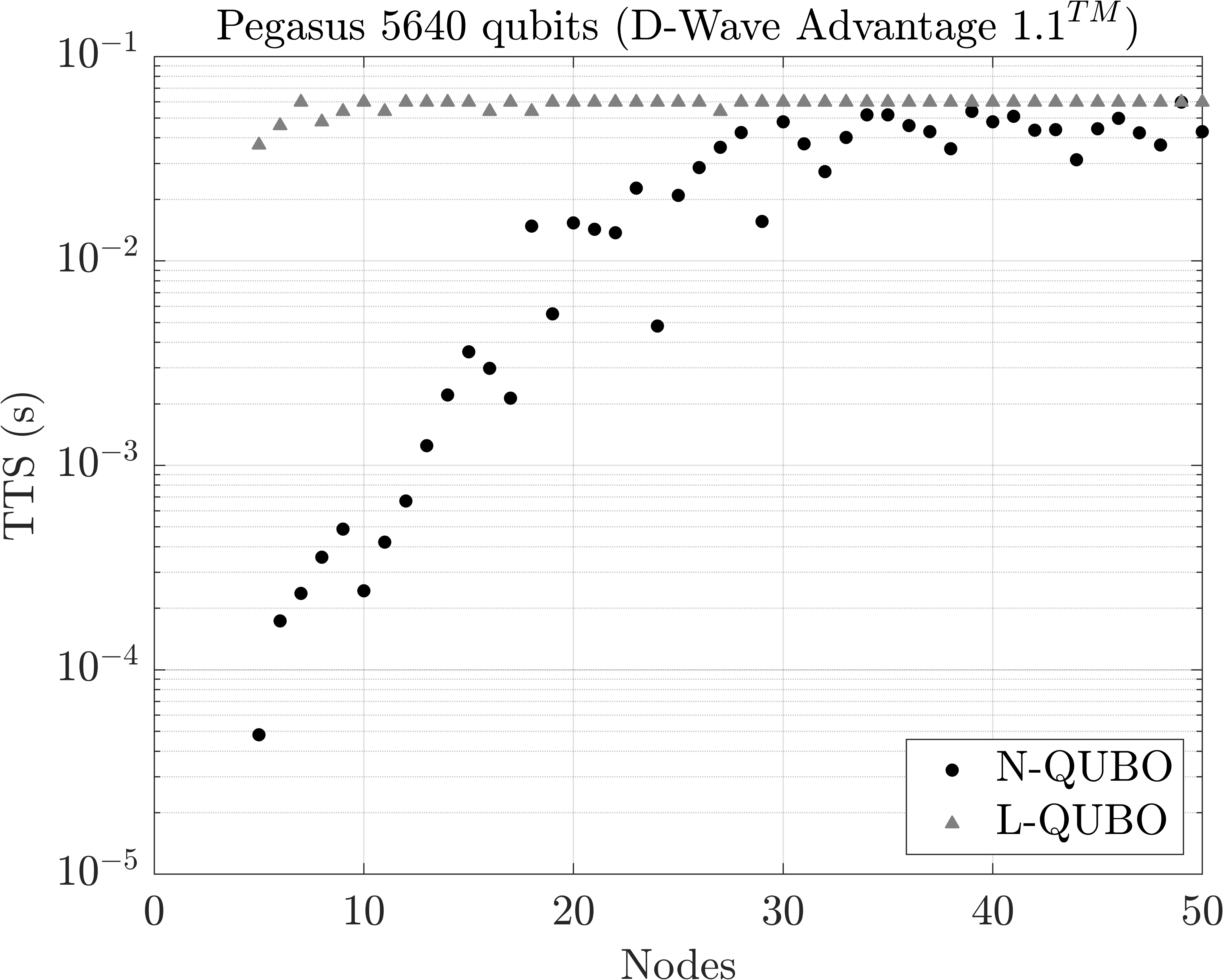}
        \caption{$k=2$, $\G(n,0.75)$, $c_5=c_2=5$. \label{fig:TTS_Peg_k2_p75_c1(5)_c2(5)}}
    \end{minipage}
\end{figure}

Not surprisingly, pairwise comparing Figures~\ref{fig:TTS_Chi_k2_p25_c1(1)_c2(1)},~\ref{fig:TTS_Chi_k2_p75_c1(1)_c2(1)},
with Figures~\ref{fig:TTS_Peg_k2_p25_c1(1)_c2(1)},~\ref{fig:TTS_Peg_k2_p75_c1(1)_c2(1)}, the advantages of the Advantage 1.1\textsuperscript{TM} processor over the 2000Q\textsuperscript{TM} processor in terms of TTS are clear. For the L-QUBO the newer processor finds solutions for much larger instances of the M$k$CS problem. For the N-QUBO, the rate of increase of the TTS as the size of the M$k$CS problem increases is about one order of magnitude lower in the newer processor.

We can also use the results presented in this question to study whether the conclusions made in Section~\ref{sec:mingap}, particularly in Table~\ref{tab:deltas}, reflect on the actual TTS time when using a quantum annealer. Note that from Table~\ref{tab:deltas} it was expected that for instances of the M$k$CS problem with $k=2$ and underlying graphs $\G(n,0.25)$, increasing the penalty constants from $c_1=c_2 =1$ to $c_1=c_2 =5$ would result in a faster convergence. However, by comparing Figures~\ref{fig:TTS_Chi_k2_p25_c1(1)_c2(1)} and~\ref{fig:TTS_Chi_k2_p25_c1(5)_c2(5)}, as well as  Figures~\ref{fig:TTS_Peg_k2_p25_c1(1)_c2(1)} and~\ref{fig:TTS_Peg_k2_p25_c1(5)_c2(5)}, it follows that increasing the penalty constants in this way is actually counterproductive for both quantum annealing processors in terms of TTS (i.e., in Figures~\ref{fig:TTS_Chi_k2_p25_c1(5)_c2(5)} and~\ref{fig:TTS_Peg_k2_p25_c1(5)_c2(5)}, the ``slope'' at which the TTS increases with the number of nodes is higher. Also,  from Table~\ref{tab:deltas} it was expected that for instances of the M$k$CS problem with $k=2$ and underlying graphs $\G(n,0.75)$, the increase in the penalty constants from $c_1=c_2 =1$ to $c_1=c_2 =5$ would have an even slightly higher benefit in terms of speed of convergence (compared with  $\G(n,0.25)$ graphs).  However, by comparing Figures~\ref{fig:TTS_Chi_k2_p75_c1(1)_c2(1)} and~\ref{fig:TTS_Chi_k2_p75_c1(5)_c2(5)}, as well as  Figures~\ref{fig:TTS_Peg_k2_p75_c1(1)_c2(1)} and~\ref{fig:TTS_Peg_k2_p75_c1(5)_c2(5)}, it follows that increasing the penalty constants in this way does not produce discernible improvements for the quantum annealing processors in terms of TTS.
Most likely, this means that any theoretical advantages in terms of convergence obtained by increasing the value of penalty constants is off-set by the precision problems that using larger penalty parameters brings for the quantum annealing processors in practice. The fact that being able to use penalty parameters close to one (1) is beneficial for quantum annealers is discussed, for example, in~\citep{vyskovcil2019embedding}. This shows the importance of being able to fully characterize the range of penalty constants that result in a QUBO being a reformulation of a COPT problem.


%
%

\section{Concluding remarks}
\label{sec:end}
In this paper, we consider a particularly important combinatorial optimization (COPT) problem; namely, the maximum $k$-colorable subgraph (M$k$CS) problem, in which the aim  is to find an induced $k$-colorable subgraph
with maximum cardinality in a given graph. 
This problem arises in channel assignment in spectrum sharing networks (e.g., Wi-Fi or cellular), VLSI design, human genetic research, cybersecurity, cryptography, and scheduling.
We derive two QUBO reformulations for the M$k$CS problem; a linear-based QUBO reformulation (Theorem~\ref{thm:linearQUBO}) and a nonlinear-based QUBO reformulation (Theorem~\ref{thm:nonlinearQUBO}). 
Furthermore, we 
fully characterize the range of the penalty parameters that can be used in the QUBO reformulation. In the case of the linear-based QUBO reformulation, this analysis shows that Theorem~\ref{thm:linearQUBO} provides a better QUBO reformulation for the M$k$CS problem than the one that could be obtained using the QUBO reformulation techniques recently introduced by~\citet{lasserre2016max}. In the case of the nonlinear-based QUBO reformulation, this analysis shows that Theorem~\ref{thm:nonlinearQUBO} provides a better QUBO reformulation for the M$k$CS problem than the one that could be obtained using the well-known QUBO reformulation techniques introduced by~\citet{lucas2014ising}. Our proofs bring forward a fact that is overlooked in related articles. Namely, that when minimal penalty parameters are used in QUBO reformulations, the equivalence in terms of objective value between a problem and its associated QUBO reformulation does not necessarily mean that the optimal solution of the QUBO reformulation provides a feasible, optimal solution for the original problem. This is shown to be the case for the M$k$CS problem in Corollaries~\ref{cor:unitnonlinear} and~\ref{cor:unitlinear}. Given that for $k=1$ the M$k$CS problem is equivalent to the stable set problem, we show (see end of Section~\ref{sec:nonlinear}) that  this issue applies to the well-known QUBO reformulation of the stable set problem~\eqref{eq:alphaQUBO}. However, we show that this issue can be be simply addressed by using the greedy Algorithm~\ref{alg:AlgorithmA} (for general instances of the M$k$CS problem).

We finish in Section~\ref{sec:benchmark} by illustrating the advantages of the nonlinear-based QUBO reformulation over the linear-based QUBO reformulation in terms of embedding requirements, convergence rate, and time-to-solution when the QUBO reformulations are used to solve the M$k$CS problem in a quantum annealing device. The experiments also illustrate the importance of having a full characterization of the penalty parameters that ensure the proposed QUBOs are indeed reformulations of the original problem. For example, we explore the potential theoretical and practical gains of using higher penalty parameters than the minimum ones required for the QUBO to become a reformulation of the M$k$CS problem. Our results show that although there are some theoretical benefits of using larger than minimal penalty parameters, they do not translate to a faster convergence to a solution of the problem on a quantum annealing computing device.

Our results contribute to recent literature that beyond obtaining QUBO reformulations of COPT problems such as the graph isomorphism problem as well as tree and cycle elimination problems, look for {\em improved} QUBO reformulations of these problems for NISQ devices~\citep[see, e.g.,][]{calude2017, hua2020improved, fowler2017improved, verma2020optimal, verma2020penalty}. That is, QUBO reformulations that are tailored to be more efficiently used in NISQ devices.

\section*{Acknowledgements} This project has been carried out thanks to funding by the Defense Advanced Research Projects Agency (DARPA), ONISQ grant W911NF2010022, titled {\em The Quantum
Computing Revolution and Optimization: Challenges and Opportunities}. The project was also supported by the Oak Ridge National Laboratory OLCF grant ENG121, which provided the authors with in-kind access to D-Wave's quantum annealers. The second author acknowledges the support of the Center for Advanced Process Decision Making (CAPD) at Carnegie Mellon University.

\bibliographystyle{apalike}
\bibliography{QUBO_BibTeX}

\begin{thebibliography}{}

\bibitem[Abello et~al., 2001]{abello2001finding}
Abello, J., Butenko, S., Pardalos, P.~M., and Resende, M.~G. (2001).
\newblock Finding independent sets in a graph using continuous multivariable
  polynomial formulations.
\newblock {\em Journal of Global Optimization}, 21(2):111--137.

\bibitem[Amin, 2008]{amin2008effect}
Amin, M. (2008).
\newblock Effect of local minima on adiabatic quantum optimization.
\newblock {\em Physical Review Letters}, 100(13):130503.

\bibitem[Amin et~al., 2012]{amin2012approximate}
Amin, M.~H., Smirnov, A.~Y., Dickson, N.~G., and Drew-Brook, M. (2012).
\newblock Approximate diagonalization method for large-scale hamiltonians.
\newblock {\em Physical Review A}, 86(5):052314.

\bibitem[Berman and Pelc, 1990]{berman1990distributed}
Berman, P. and Pelc, A. (1990).
\newblock Distributed probabilistic fault diagnosis for multiprocessor systems.
\newblock In {\em [1990] Digest of Papers. Fault-Tolerant Computing: 20th
  International Symposium}, pages 340--346. IEEE.

\bibitem[Bernal et~al., 2020]{bernal2020integer}
Bernal, D.~E., Booth, K.~E., Dridi, R., Alghassi, H., Tayur, S., and
  Venturelli, D. (2020).
\newblock Integer programming techniques for minor-embedding in quantum
  annealers.
\newblock In {\em International Conference on Integration of Constraint
  Programming, Artificial Intelligence, and Operations Research}, pages
  112--129. Springer.

\bibitem[Bomze et~al., 1999]{pardalos1994maximum}
Bomze, I.~M., Budinich, M., Pardalos, P.~M., and Pelillo, M. (1999).
\newblock The maximum clique problem.
\newblock In Du, D.~Z. and Pardalos, P.~M., editors, {\em Handbook of
  {C}ombinatorial {O}ptimization}, pages 1--74. Kluwer Academic Publisher.

\bibitem[Boothby et~al., 2020]{boothby2020next}
Boothby, K., Bunyk, P., Raymond, J., and Roy, A. (2020).
\newblock Next-generation topology of {D}-{W}ave quantum processors.
\newblock {\em arXiv preprint arXiv:2003.00133}.

\bibitem[Boros et~al., 2007]{boros2007local}
Boros, E., Hammer, P.~L., and Tavares, G. (2007).
\newblock Local search heuristics for quadratic unconstrained binary
  optimization ({QUBO}).
\newblock {\em Journal of Heuristics}, 13(2):99--132.

\bibitem[Brush, 1967]{brush1967history}
Brush, S.~G. (1967).
\newblock History of the {L}enz-{I}sing model.
\newblock {\em Reviews of Modern Physics}, 39(4):883.

\bibitem[Calude et~al., 2017]{calude2017}
Calude, C.~S., Dinneen, M.~J., and Hua, R. (2017).
\newblock {QUBO} formulations for the graph isomorphism problem and related
  problems.
\newblock {\em Theoretical Computer Science}, 701:54--69.

\bibitem[Camp{\^e}lo and Corr{\^e}a, 2010]{campelo2010combined}
Camp{\^e}lo, M. and Corr{\^e}a, R.~C. (2010).
\newblock A combined parallel lagrangian decomposition and cutting-plane
  generation for maximum stable set problems.
\newblock {\em Electronic Notes in Discrete Mathematics}, 36:503--510.

\bibitem[Chapuis et~al., 2017]{chapuis2017finding}
Chapuis, G., Djidjev, H., Hahn, G., and Rizk, G. (2017).
\newblock Finding maximum cliques on a quantum annealer.
\newblock In {\em Proceedings of the Computing Frontiers Conference}, pages
  63--70.

\bibitem[Choi, 2008]{choi2008minor}
Choi, V. (2008).
\newblock Minor-embedding in adiabatic quantum computation: I. {T}he parameter
  setting problem.
\newblock {\em Quantum Information Processing}, 7(5):193--209.

\bibitem[Cipra, 2000]{cipra2000ising}
Cipra, B.~A. (2000).
\newblock The {I}sing model is {NP}-complete.
\newblock {\em SIAM News}, 33(6):1--3.

\bibitem[Cole, 2018]{press}
Cole, S. (Mar. 20, 2018).
\newblock Ready or not, the quantum computing revolution is here.
\newblock {\em Military embedded systems}.
\newblock
  \url{http://mil-embedded.com/articles/ready-not-quantum-computing-revolution-here/}.

\bibitem[Conforti et~al., 2014]{conforti2014integer}
Conforti, M., Cornu{\'e}jols, G., Zambelli, G., et~al. (2014).
\newblock {\em Integer programming}, volume 271.
\newblock Springer.

\bibitem[Crooks, 2018]{crooks2018performance}
Crooks, G.~E. (2018).
\newblock Performance of the quantum approximate optimization algorithm on the
  maximum cut problem.
\newblock {\em arXiv preprint arXiv:1811.08419}.

\bibitem[Cullimore et~al., 2012]{cullimore2012relationship}
Cullimore, M., Everitt, M.~J., Ormerod, M., Samson, J., Wilson, R.~D., and
  Zagoskin, A.~M. (2012).
\newblock Relationship between minimum gap and success probability in adiabatic
  quantum computing.
\newblock {\em Journal of Physics A: Mathematical and Theoretical},
  45(50):505305.

\bibitem[{D-Wave}, 2021]{DWAVEanneal}
{D-Wave} (last accessed 01/19/2021).
\newblock {QPU-S}pecific {A}nneal {S}chedules.
\newblock
  \url{https://support.dwavesys.com/hc/en-us/articles/360005267253-QPU-Specific-Anneal-Schedules}.
\newblock Accessed: 2020-01-09.

\bibitem[Farhi et~al., 2014]{farhi2014quantum}
Farhi, E., Goldstone, J., and Gutmann, S. (2014).
\newblock A quantum approximate optimization algorithm.
\newblock {\em arXiv preprint arXiv:1411.4028}.

\bibitem[Farhi et~al., 2001]{farhi2001quantum}
Farhi, E., Goldstone, J., Gutmann, S., Lapan, J., Lundgren, A., and Preda, D.
  (2001).
\newblock A quantum adiabatic evolution algorithm applied to random instances
  of an {NP}-complete problem.
\newblock {\em Science}, 292(5516):472--475.

\bibitem[Farhi et~al., 2000]{farhi2000quantum}
Farhi, E., Goldstone, J., Gutmann, S., and Sipser, M. (2000).
\newblock Quantum computation by adiabatic evolution.
\newblock {\em arXiv preprint quant-ph/0001106}.

\bibitem[Farhi and Harrow, 2016]{farhi2016quantum}
Farhi, E. and Harrow, A.~W. (2016).
\newblock Quantum supremacy through the quantum approximate optimization
  algorithm.
\newblock {\em arXiv preprint arXiv:1602.07674}.

\bibitem[Fouilhoux and Mahjoub, 2012]{fouilhoux2012solving}
Fouilhoux, P. and Mahjoub, A.~R. (2012).
\newblock Solving {VLSI} design and {DNA} sequencing problems using
  bipartization of graphs.
\newblock {\em Computational Optimization and Applications}, 51(2):749--781.

\bibitem[Fowler, 2017]{fowler2017improved}
Fowler, A. (2017).
\newblock {\em Improved {QUBO} formulations for {D}-{W}ave quantum computing}.
\newblock PhD thesis, University of Auckland.

\bibitem[Glover et~al., 2019]{glover2019quantum}
Glover, F., Kochenberger, G., and Du, Y. (2019).
\newblock Quantum bridge analytics {I}: a tutorial on formulating and using
  {QUBO} models.
\newblock {\em 4OR}, 17(4):335--371.

\bibitem[Goemans and Williamson, 1995]{goemans1995improved}
Goemans, M.~X. and Williamson, D.~P. (1995).
\newblock Improved approximation algorithms for maximum cut and satisfiability
  problems using semidefinite programming.
\newblock {\em Journal of the ACM (JACM)}, 42(6):1115--1145.

\bibitem[Guerreschi and Matsuura, 2019]{guerreschi2019qaoa}
Guerreschi, G.~G. and Matsuura, A.~Y. (2019).
\newblock {QAOA} for max-cut requires hundreds of qubits for quantum speed-up.
\newblock {\em Scientific reports}, 9(1):6903.

\bibitem[Halld{\'o}rsson et~al., 2010]{halldorsson2010spectrum}
Halld{\'o}rsson, M.~M., Halpern, J.~Y., Li, L.~E., and Mirrokni, V.~S. (2010).
\newblock On spectrum sharing games.
\newblock {\em Distributed computing}, 22(4):235--248.

\bibitem[Harant, 2000]{harant2000some}
Harant, J. (2000).
\newblock Some news about the independence number of a graph.
\newblock {\em Discussiones Mathematicae Graph Theory}, 20(1):71--79.

\bibitem[Hua and Dinneen, 2020]{hua2020improved}
Hua, R. and Dinneen, M.~J. (2020).
\newblock Improved {QUBO} formulation of the graph isomorphism problem.
\newblock {\em SN Computer Science}, 1(1):19.

\bibitem[Januschowski and Pfetsch, 2011]{januschowski2011maximum}
Januschowski, T. and Pfetsch, M.~E. (2011).
\newblock The maximum k-colorable subgraph problem and orbitopes.
\newblock {\em Discrete Optimization}, 8(3):478--494.

\bibitem[Johnson et~al., 2011]{johnson2011quantum}
Johnson, M.~W., Amin, M.~H., Gildert, S., Lanting, T., Hamze, F., Dickson, N.,
  Harris, R., Berkley, A.~J., Johansson, J., Bunyk, P., et~al. (2011).
\newblock Quantum annealing with manufactured spins.
\newblock {\em Nature}, 473(7346):194--198.

\bibitem[Karp, 1972]{karp1972reducibility}
Karp, R.~M. (1972).
\newblock Reducibility among combinatorial problems.
\newblock In {\em Complexity of computer computations}, pages 85--103.
  Springer.

\bibitem[King and Bernoudy, 2020]{king2020performance}
King, A.~D. and Bernoudy, W. (2020).
\newblock Performance benefits of increased qubit connectivity in quantum
  annealing 3-dimensional spin glasses.
\newblock {\em arXiv preprint arXiv:2009.12479}.

\bibitem[King and McGeoch, 2014]{king2014algorithm}
King, A.~D. and McGeoch, C.~C. (2014).
\newblock Algorithm engineering for a quantum annealing platform.
\newblock {\em arXiv preprint arXiv:1410.2628}.

\bibitem[King et~al., 2015]{king2015benchmarking}
King, J., Yarkoni, S., Nevisi, M.~M., Hilton, J.~P., and McGeoch, C.~C. (2015).
\newblock Benchmarking a quantum annealing processor with the time-to-target
  metric.
\newblock {\em arXiv preprint arXiv:1508.05087}.

\bibitem[Kuryatnikova et~al., 2020]{kuryatnikova2020maximum}
Kuryatnikova, O., Sotirov, R., and Vera, J. (2020).
\newblock The maximum $ k $-colorable subgraph problem and related problems.
\newblock {\em arXiv preprint arXiv:2001.09644}.

\bibitem[Lasserre, 2016]{lasserre2016max}
Lasserre, J.~B. (2016).
\newblock A max-cut formulation of 0/1 programs.
\newblock {\em Operations Research Letters}, 44(2):158--164.

\bibitem[Lippert et~al., 2002]{lippert2002algorithmic}
Lippert, R., Schwartz, R., Lancia, G., and Istrail, S. (2002).
\newblock Algorithmic strategies for the single nucleotide polymorphism
  haplotype assembly problem.
\newblock {\em Briefings in bioinformatics}, 3(1):23--31.

\bibitem[Lov{\'a}sz, 1979]{lovasz1979shannon}
Lov{\'a}sz, L. (1979).
\newblock On the {S}hannon capacity of a graph.
\newblock {\em IEEE Transactions on Information theory}, 25(1):1--7.

\bibitem[Lucas, 2014]{lucas2014ising}
Lucas, A. (2014).
\newblock {I}sing formulations of many {NP} problems.
\newblock {\em Frontiers in Physics}, 2:5.

\bibitem[Lund and Yannakakis, 1993]{lund1993approximation}
Lund, C. and Yannakakis, M. (1993).
\newblock The approximation of maximum subgraph problems.
\newblock In {\em International Colloquium on Automata, Languages, and
  Programming}, pages 40--51. Springer.

\bibitem[Metz, 2018]{nytimes}
Metz, C. (Oct. 21, 2018).
\newblock The next tech talent shortage: {Q}uantum computing researchers.
\newblock {\em New York Times}.
\newblock
  \url{https://www.nytimes.com/2018/10/21/technology/quantum-computing-jobs-immigration-visas.html}.

\bibitem[Montanaro, 2016]{montanaro2016quantum}
Montanaro, A. (2016).
\newblock Quantum algorithms: an overview.
\newblock {\em npj Quantum Information}, 2(1):1--8.

\bibitem[Nannicini, 2019]{nannicini2019performance}
Nannicini, G. (2019).
\newblock Performance of hybrid quantum-classical variational heuristics for
  combinatorial optimization.
\newblock {\em Physical Review E}, 99(1):013304.

\bibitem[Neven et~al., 2009]{neven2009nips}
Neven, H., Denchev, V.~S., Drew-Brook, M., Zhang, J., Macready, W.~G., and
  Rose, G. (2009).
\newblock {NIPS} 2009 demonstration: Binary classification using hardware
  implementation of quantum annealing.
\newblock {\em Quantum}, pages 1--17.

\bibitem[Pajouh et~al., 2013]{pajouh2013characterization}
Pajouh, F.~M., Balasundaram, B., and Prokopyev, O.~A. (2013).
\newblock On characterization of maximal independent sets via quadratic
  optimization.
\newblock {\em Journal of Heuristics}, 19(4):629--644.

\bibitem[Papalitsas et~al., 2019]{papalitsas2019qubo}
Papalitsas, C., Andronikos, T., Giannakis, K., Theocharopoulou, G., and
  Fanarioti, S. (2019).
\newblock A {QUBO} model for the traveling salesman problem with time windows.
\newblock {\em Algorithms}, 12(11):224.

\bibitem[Poljak and Tuza, 1995]{poljak1995maximum}
Poljak, S. and Tuza, Z. (1995).
\newblock Maximum cuts and large bipartite subgraphs.
\newblock {\em DIMACS Series}, 20:181--244.

\bibitem[Rieffel et~al., 2015]{rieffel2015case}
Rieffel, E.~G., Venturelli, D., O'Gorman, B., Do, M.~B., Prystay, E.~M., and
  Smelyanskiy, V.~N. (2015).
\newblock A case study in programming a quantum annealer for hard operational
  planning problems.
\newblock {\em Quantum Information Processing}, 14(1):1--36.

\bibitem[Roland and Cerf, 2002]{roland2002quantum}
Roland, J. and Cerf, N.~J. (2002).
\newblock Quantum search by local adiabatic evolution.
\newblock {\em Physical Review A}, 65(4):042308.

\bibitem[R{\o}nnow et~al., 2014]{ronnow2014defining}
R{\o}nnow, T.~F., Wang, Z., Job, J., Boixo, S., Isakov, S.~V., Wecker, D.,
  Martinis, J.~M., Lidar, D.~A., and Troyer, M. (2014).
\newblock Defining and detecting quantum speedup.
\newblock {\em science}, 345(6195):420--424.

\bibitem[Ruan et~al., 2020]{ruan2020quantum}
Ruan, Y., Marsh, S., Xue, X., Li, X., Liu, Z., and Wang, J. (2020).
\newblock Quantum approximate algorithm for np optimization problems with
  constraints.
\newblock {\em arXiv preprint arXiv:2002.00943}.

\bibitem[Semeniuk, 2017]{canpress}
Semeniuk, I. (Sept., 2017).
\newblock Understanding the {Q}uantum {C}omputing {R}evolution.
\newblock {\em The {G}lobe and {M}ail}.
\newblock
  \url{https://www.theglobeandmail.com/report-on-business/rob-magazine/quantum-computing-technology-explained/article36397793/}.

\bibitem[Singh, 2020]{singh2020ising}
Singh, S.~P. (2020).
\newblock The {I}sing model: Brief introduction and its application.
\newblock In {\em Solid State Physics-Metastable, Spintronics Materials and
  Mechanics of Deformable Bodies-Recent Progress}. IntechOpen.

\bibitem[Stollenwerk et~al., 2019]{stollenwerk2019quantum}
Stollenwerk, T., O?Gorman, B., Venturelli, D., Mandr{\`a}, S., Rodionova, O.,
  Ng, H., Sridhar, B., Rieffel, E.~G., and Biswas, R. (2019).
\newblock Quantum annealing applied to de-conflicting optimal trajectories for
  air traffic management.
\newblock {\em IEEE transactions on intelligent transportation systems},
  21(1):285--297.

\bibitem[Subramanian et~al., 2007]{subramanian2007fast}
Subramanian, A.~P., Gupta, H., Das, S.~R., and Buddhikot, M.~M. (2007).
\newblock Fast spectrum allocation in coordinated dynamic spectrum access based
  cellular networks.
\newblock In {\em 2007 2nd IEEE International Symposium on New Frontiers in
  Dynamic Spectrum Access Networks}, pages 320--330. IEEE.

\bibitem[van Dam and Sotirov, 2016]{van2016new}
van Dam, E.~R. and Sotirov, R. (2016).
\newblock New bounds for the max-k-cut and chromatic number of a graph.
\newblock {\em Linear Algebra and its Applications}, 488:216--234.

\bibitem[Venturelli et~al., 2016]{venturelli2016job}
Venturelli, D., Marchand, D., and Rojo, G. (2016).
\newblock Job shop scheduling solver based on quantum annealing.
\newblock In {\em Proc. of ICAPS-16 Workshop on Constraint Satisfaction
  Techniques for Planning and Scheduling (COPLAS)}, pages 25--34.

\bibitem[Verma and Lewis, 2020a]{verma2020optimal}
Verma, A. and Lewis, M. (2020a).
\newblock Optimal quadratic reformulations of fourth degree pseudo-boolean
  functions.
\newblock {\em Optimization Letters}, 14(6):1557--1569.

\bibitem[Verma and Lewis, 2020b]{verma2020penalty}
Verma, A. and Lewis, M. (2020b).
\newblock Penalty and partitioning techniques to improve performance of {QUBO}
  solvers.
\newblock {\em Discrete Optimization}, page 100594.

\bibitem[Vysko{\v{c}}il et~al., 2019]{vyskovcil2019embedding}
Vysko{\v{c}}il, T., Pakin, S., and Djidjev, H.~N. (2019).
\newblock Embedding inequality constraints for quantum annealing optimization.
\newblock In {\em International workshop on quantum technology and optimization
  problems}, pages 11--22. Springer.

\bibitem[Wang et~al., 2018]{wang2018quantum}
Wang, Z., Hadfield, S., Jiang, Z., and Rieffel, E.~G. (2018).
\newblock Quantum approximate optimization algorithm for maxcut: A fermionic
  view.
\newblock {\em Physical Review A}, 97(2):022304.

\bibitem[Wocjan and Beth, 2003]{wocjan20032}
Wocjan, P. and Beth, T. (2003).
\newblock The 2-local hamiltonian problem encompasses {NP}.
\newblock {\em International Journal of Quantum Information}, 1(03):349--357.

\bibitem[Yannakakis and Gavril, 1987]{yannakakis1987maximum}
Yannakakis, M. and Gavril, F. (1987).
\newblock The maximum k-colorable subgraph problem for chordal graphs.
\newblock {\em Information Processing Letters}, 24(2):133--137.

\bibitem[Yarkoni et~al., 2018]{yarkoni2018first}
Yarkoni, S., Plaat, A., and Back, T. (2018).
\newblock First results solving arbitrarily structured maximum independent set
  problems using quantum annealing.
\newblock In {\em 2018 IEEE Congress on Evolutionary Computation (CEC)}, pages
  1--6. IEEE.

\bibitem[Zbinden et~al., 2020]{zbinden2020embedding}
Zbinden, S., B{\"a}rtschi, A., Djidjev, H., and Eidenbenz, S. (2020).
\newblock Embedding algorithms for quantum annealers with chimera and pegasus
  connection topologies.
\newblock In {\em International Conference on High Performance Computing},
  pages 187--206. Springer.

\end{thebibliography}

\end{document}